\documentclass[runningheads]{llncs}
\usepackage[normalem]{ulem}
\usepackage[]{algorithm2e}
\SetKwInput{KwData}{Input}
\SetKwInput{KwResult}{Output}

\usepackage{soul} % to be able to cross text 
\usepackage{tikz}
\usetikzlibrary{plotmarks}
\usepackage{pgfplots,pgfplotstable}

\usepackage{hyperref}

\usepackage{arydshln}

\usepackage{amsmath,amssymb,amsfonts}
\usepackage{color} % for textcolor (and maybe more)
\newtheorem{notation}{Notation}
\newtheorem{modelling}{Modeling}
\newtheorem{heuristic}{Heuristic}

\newcommand{\ff}[1]{\mathbb{F}_{#1}} % finite field

\newcommand{\Fq}{\ff{q}}
\newcommand{\Fqm}{\ff{q^m}}
\newcommand{\NN}{\mathbb{N}} % nonnegative integers

\newcommand{\Ceil}[1]{\left\lceil #1 \right\rceil}
\newcommand{\Ceiling}[1]{\Ceil{#1}}
\newcommand{\minor}[2]{\left| #1 \right|_{#2}}
\newcommand{\size}[1]{\#{#1}}

\newcommand{\any}{*} % any entry of matrix, typically for submatrices
\newcommand{\rank}[1]{\operatorname{Rank}\mathchoice{\left(#1\right)}{(#1)}{(#1)}{(#1)}} % rank of a word / of a matrix
\newcommand{\rw}[1]{\left| #1 \right|_{\textsc{rank}}}
\DeclareMathOperator{\Mat}{Mat}

\DeclareMathOperator{\MaxMinors}{{\bf MaxMinors}}
\DeclareMathOperator{\MaxMin}{{\bf MaxMin}}

\DeclareMathOperator{\Unfold}{{\bf UnFold}}

\newcommand{\Sign}[2]{\sigma_{#2}{(#1)}}
\newcommand{\matRing}[3]{#1^{#2 \times #3}} % matrix ring: #1 is the base ring/field, #2 the row dimension, #3 the column dimension
\newcommand{\mat}[1]{\boldsymbol{#1}} % font for matrix
\newcommand{\word}[1]{\vec{\boldsymbol{#1}}} % font for vector
\newcommand{\zerov}{\word{0}} % zero vector
\newcommand{\trsp}[1]{#1^\intercal} % transpose matrix
\newcommand{\cv}{\mat{c}}

\newcommand{\ev}{\mat{e}}
\newcommand{\rv}{\mat{r}}
\newcommand{\uv}{\mat{u}}
\newcommand{\vv}{\mat{v}}
\newcommand{\xv}{\mat{x}}
\newcommand{\yv}{\mat{y}}
\newcommand{\zerom}{\mat{0}}
\newcommand{\Am}{\mat{A}}

\newcommand{\Cm}{\mat{C}}

\newcommand{\Hm}{\mat{H}}
\renewcommand{\Im}{\mat{I}}

\newcommand{\Mm}{\mat{M}}

\newcommand{\Rm}{\mat{R}}
\newcommand{\Sm}{\mat{S}}

\newcommand{\Cc}{{\mathcal C}}
\newcommand{\Cy}{\widetilde{C}}

\newcommand{\eqdef}{:=} % TODO a discuter

\newcommand{\Th}[1]{\Theta\left( #1 \right)}

\newcommand{\Exp}{{\mathrm{D_{exp}}}}% for the Exp term in section 5

\newcommand{\cm}{\mat{C}} % matrix "coordinate of e"
\usepackage[capitalize]{cleveref}

 \makeatletter
\newcommand\bibalias[2]{%
  \@namedef{bibali@#1}{#2}%
}

\newtoks\biba@toks
\newcommand\acite[2][]{%
  \biba@toks{\cite#1}%
  \def\biba@comma{}%
  \def\biba@all{}%
  \@for\biba@one:=#2\do{%
    \@ifundefined{bibali@\biba@one}{%
      \edef\biba@all{\biba@all\biba@comma\biba@one}%
    }{%
      \PackageInfo{bibalias}{%
        Replacing citation `\biba@one' with `\@nameuse{bibali@\biba@one}'
      }%
      \edef\biba@all{\biba@all\biba@comma\@nameuse{bibali@\biba@one}}%
    }%
    \def\biba@comma{,}%
  }%
  \edef\biba@tmp{\the\biba@toks{\biba@all}}%
  \biba@tmp
}
\makeatother

\bibalias{Ouroboros-R}{AABBBDGHZ17}
\bibalias{RQC}{AABBBDGZ17}
\bibalias{RQC2}{AABBBDGZCH19}
\bibalias{ROLLO}{ABDGHRTZABBBO19}
\bibalias{LAKE}{ABDGHRTZ17}
\bibalias{LOCKER}{ABDGHRTZ17a}
\bibalias{RankSign}{AGHRZ17}

\newcommand{\jp}[1]{\textcolor{red}{\bf [JP: #1]}}
\newcommand{\bros}[1]{\textcolor{blue}{\bf [Max: #1]}}

\renewcommand{\jp}[1]{\textcolor{red}{}}
\renewcommand{\bros}[1]{\textcolor{blue}{}}

\begin{document}

\title{ Improvements of Algebraic Attacks 
        for solving the Rank Decoding and MinRank problems
      }

%\titlerunning{Abbreviated paper title}
% If the paper title is too long for the running head, you can set
% an abbreviated paper title here

\author{}
\institute{}
\author{
  Magali Bardet\inst{4,5} \and
  Maxime Bros\inst{1} \and
  Daniel Cabarcas\inst{6} \and
  Philippe Gaborit\inst{1} \and
  Ray Perlner\inst{2} \and
  Daniel Smith-Tone\inst{2,3} \and
  Jean-Pierre Tillich\inst{4} \and
  Javier Verbel\inst{6}
}

% \authorrunning{F. Author et al.}
% First names are abbreviated in the running head.
% If there are more than two authors, 'et al.' is used.

\institute{
  Univ. Limoges, CNRS, XLIM, UMR 7252, F-87000 Limoges, France \\
  \email{maxime.bros@unilim.fr}
  \and
  National Institute of Standards and Technology, USA
  \and
  University of Louisville, USA
  \and 
  Inria, 2 rue Simone Iff, 75012 Paris, France
  \and
  LITIS, University of Rouen Normandie, France
  \and
  Universidad Nacional de Colombia Sede Medellín, Medellín, Colombia
}

\maketitle
\begin{abstract}
In this paper, we show how to significantly improve 
algebraic techniques for solving the MinRank problem, which is 
ubiquitous in multivariate and rank metric code based cryptography.
In the case of the structured MinRank instances arising in the latter, 
we build upon a recent breakthrough \cite{BBBGNRT19} showing that 
algebraic attacks outperform the combinatorial ones that 
were considered state of the art up until now.
Through a slight modification of this approach, we completely avoid Gr\"obner bases 
computations for certain parameters and are left only with solving linear systems.
This does not only substantially improve the complexity, but also 
gives a convincing argument as to why algebraic techniques
work in this case. When used against the second round NIST-PQC candidates 
ROLLO-I-128/192/256, our new attack has bit complexity respectively 71, 87, and 151, 
to be compared to 117, 144, and 197 as obtained in \cite{BBBGNRT19}. 
The linear systems arise from the nullity of the maximal minors of a 
certain matrix associated to the algebraic modeling. 
We also use a similar approach to improve the algebraic MinRank solvers 
for the usual MinRank problem. 
When applied against the second round NIST-PQC candidates GeMSS and Rainbow, 
our attack has a complexity that is very close to or even slightly better 
than those of the best known attacks so far. 
Note that these latter attacks did not rely on MinRank techniques  
since the MinRank approach used to give complexities that were far away from 
classical security levels.

\end{abstract}

\section{Introduction}\label{sec:intro}

      \subsubsection*{Rank metric code-based cryptography.}

      In the last decade, rank metric code-based cryptography has proved to be a
      powerful alternative to traditional code-based cryptography based on the
      Hamming metric. This thread of research started with the GPT cryptosystem
      \cite{GPT91} based on Gabidulin codes \cite{G85}, which are rank metric
      analogues of Reed-Solomon codes. However, the strong algebraic structure of
      those codes was successfully exploited for attacking the original GPT
      cryptosystem and its variants with the Overbeck attack \cite{O05} (see 
%JP for example 
\cite{OTN18} for 
%JP one of 
the latest 
%JP related 
developments). 
%JP redite This can
%has to
%      be traced back to the algebraic structure of Gabidulin codes
%JP that makes masking     extremely difficult
This is similar to the Hamming
% one can draw a parallel with the situation in the Hamming
      metric situation where essentially all McEliece cryptosystems based on Reed-Solomon codes
      or variants of them have been broken. However, recently a rank metric analogue
      of the NTRU cryptosystem \cite{HPS98} has been designed and studied,
      starting with the pioneering paper \cite{GMRZ13}. 
%JP Roughly speaking, 
NTRU
%the NTRU
%      cryptosystem 
relies on a lattice with 
%that has 
vectors of rather small Euclidean
      norm. It is precisely those vectors that allow an efficient
      decoding/deciphering process. The decryption of the cryptosystem proposed in
      \cite{GMRZ13} relies on LRPC codes with rather short vectors in the dual
      code, but this time for the rank metric. 
%These vectors are used for decoding in    the rank metric. 
This cryptosystem can also be viewed as the rank metric
      analogue of the MDPC cryptosystem \cite{MTSB12} relying on short dual code vectors for the Hamming metric. 

      This new way of building rank metric code-based cryptosystems has led to a
      sequence of proposals \acite{GMRZ13,GRSZ14,LAKE,LOCKER}, culminating in
      submissions to the NIST post-quantum competition \acite{Ouroboros-R,RQC}, whose
      security relies solely on decoding codes  in rank metric 
%codes 
with a ring
      structure similar to those used in lattice-based
      cryptography. Interestingly enough, one can also build signature schemes using
      the rank metric; even though early attempts which relied on masking the
      structure of a code \acite{GRSZ14a,RankSign} have been broken \cite{DT18b}, a
      promising recent approach \cite{ABGHZ19} only considers random matrices without
      structural masking.

      \subsubsection*{Decoding $\Fqm$-linear codes in Rank metric.}
      In other words, in rank metric code-based cryptography we are now only left
      with assessing the difficulty of the decoding problem in rank metric.
      The trend there 
%in rank metric code-based cryptography 
is to consider
%      a particular form of codes that are 
linear codes of length $n$
      over an extension $\Fqm$ of degree $m$ of $\Fq$, i.e., $\Fqm$-linear
      subspaces of $\Fqm^n$. Let  $(\beta_1,\dots,\beta_m)$ be any
      basis of $\Fqm$ as a $\Fq$-vector space. Then words 
      of those codes can be interpreted as matrices with entries in the
      ground field $\ff{q}$ by viewing a vector
      $\xv =(x_1,\dots,x_n) \in \Fqm^n$ as a matrix
      \(\Mat(\xv) = (X_{ij})_{i,j}\) in \(\matRing{\Fq}{m}{n}\), where
      $(X_{ij})_{1 \leq i \leq m}$ is the column vector formed by the
      coordinates of $x_j$ in 
%the basis
      $(\beta_1,\dots,\beta_m)$,
      i.e., \(x_j = \beta_1X_{1j}  + \cdots + \beta_mX_{mj}\).
      Then the ``rank'' metric $d$ on $\Fqm^n$ is the rank metric on the associated
      matrix space, namely
      \[
        d(\xv,\yv) \eqdef \rw{\yv - \xv},
        \quad\text{where we define }
        \rw{\xv} \eqdef \rank{\Mat(\xv)}.
      \]
      Hereafter, we will use the following terminology.
      \begin{problem}[$(m,n,k,r)$-decoding problem]
        \label{problem_RD} \\
        \indent\emph{Input}: an $\Fqm$-basis $(\cv_1,\dots,\cv_k)$ of a subspace $\Cc$ of
        $\Fqm^n$, an integer $r \in \NN$, and a vector $\yv \in \Fqm^n$ 
%at distance        at most \(r\) of \(\Cc\) (i.e.~
such that $\rw{\yv-\cv} \leq r$ for some $\cv \in \Cc$. \\
        \indent\emph{Output}: $\cv \in \Cc$ and $\ev \in \Fqm^n$  such that $\yv=\cv+\ev$
        and $\rw{\ev} \leq r$.
        \end{problem}
      This problem is known as the Rank Decoding problem, written RD. 
      It is equivalent to the Rank Syndrome 
      Decoding problem, 
%written RSD, 
for which one uses the parity check matrix of the code. 
There are two approaches to solve RD instances: the combinatorial ones 
      such as 
%JP those in \cite{GRS16} and \cite{AGHT18} 
\cite{GRS16,AGHT18} and the algebraic ones.
%      such as in \cite{BBBGNRT19}. 
For some time it was thought that the combinatorial approach was the most threatening attack on such schemes especially when $q$ is small, 
%and all the parameters rank-metric submissions, 
until 
%it became apparent in 
\cite{BBBGNRT19} showed that even for $q=2$ 
     the algebraic attacks outperform the combinatorial ones. 
%Roughly speaking,  
If the conjecture made in \cite{BBBGNRT19} holds, the complexity of solving by algebraic attacks the decoding problem is of order $2^{O(r \log n)}$ with a 
     constant depending on the code rate $R=k/n$.
% JP j'enleve cela, cela obscurcit le discours and on the algebraic constant 

      Even if the decoding problem is not known to be NP-complete for these
      $\Fqm$-linear codes, there is a randomized reduction to an NP-complete problem
      \cite{GZ14} (namely to decoding in the Hamming metric).
      The region of parameters which is of interest for the NIST submissions
      corresponds to $m = \Th{n}$, $k=\Th{n}$ and $r = \Th{\sqrt{n}}$. 

      \subsubsection*{The MinRank problem.}
      The MinRank problem was first mentioned in \cite{BFS99} where its
      NP-completeness was also proven. We will consider here the homogeneous version of this problem which corresponds to
           \begin{problem}[MinRank problem]\\
        \label{problem_minrank0}
        \indent\emph{Input}: an integer $r \in \NN$ and $K$ 
        matrices $\mat{M}_1,\dots,\mat{M}_K \in \matRing{\ff{q}}{m}{n}$.\\
        \indent\emph{Output}: field elements $x_1,x_2,\dots,x_K \in \ff{q}$ that are not all zero such that
        \begin{equation*}
          %\label{eqn:minrank}
          \rank{\sum_{i=1}^{K}x_i\Mm_i} \leq r.
        \end{equation*}
      \end{problem}
       It plays a central role in public key cryptography. Many multivariate schemes are either directly based on the hardness of this problem \cite{C01} or strongly related to it as in \cite{P96,PBD14,PCYTD15} and the NIST post-quantum competition  candidates Gui \cite{DCPSY19a}, GeMSS \cite{CFMPPR19} or Rainbow \cite{DCPSY19}. It first appeared in this context as part of Kipnis-Shamir's attack \cite{KS99} against the HFE cryptosystem \cite{P96}. It is also central 
      in rank metric code-based cryptography, because the RD problem reduces to
      MinRank as explained in \cite{FLP08} and actually the best algorithms for solving this problem are really MinRank solvers taking advantage of the $\Fqm$ underlying structure as in \cite{BBBGNRT19}. However the parameter region generally differs. When the RD problem arising from rank metric schemes is treated as a MinRank problem we generally have $K=\Th{n^2}$ and $r$ is rather small  $r = \Th{\sqrt{n}}$) whereas for the multivariate cryptosystems $K = \Th{n}$ but $r$ is much bigger. 
 
      The current best known algorithms for solving the MinRank problem have exponential complexity. Many of them are obtained by an algebraic approach too consisting
   %   \subsubsection*{Algebraic attacks.}
%
%      This family of attacks consists 
in modeling the MinRank problem by an algebraic system 
%JP of multivariate polynomial equations 
and 
%then 
solving it with Gr\"obner basis techniques.
The main modelings are the Kipnis-Shamir modeling \cite{KS99} and the minors modeling \cite{FSS10}. The complexity of solving MinRank using these
      modelings has been investigated in \cite{FLP08,FSS10,VBCPS19}. In particular \cite{VBCPS19} shows that the bilinear Kipnis-Shamir modeling behaves much better than generic bilinear systems with respect to Gr\"obner basis techniques. 
      
      \subsubsection*{Our contribution.}

      Here we follow on from the approach in \cite{BBBGNRT19} and propose 
      a slightly different modeling to solve the RD problem. Roughly speaking the algebraic approach 
%JP followed by 
in \cite{BBBGNRT19} is to set up a bilinear system satisfied by the error we are looking for. This system is formed by two kinds of variables, 
the ``coefficient'' variables and the ``support'' variables. It is implicitly the modeling considered in \cite{OJ02}. The breakthrough obtained in \cite{BBBGNRT19} was to realize that
\begin{itemize}
\item the coefficient variables have to satisfy ``maximal minor'' equations: the maximal minors of a certain $r \times (n-k-1)$ matrix  (i.e. the $r \times r$ minors)
with entries being linear forms in the coefficient variables have to be equal to $0$. 
\item these maximal minors are themselves linear combinations of maximal minors $c_T$ of an $r \times n$ matrix $\Cm$ whose entries are the coefficient variables.
\end{itemize}
This gives a linear system 
%involving the $c_T$'s which reveals 
in the  $c_T$'s provided there are enough linear equations. 
Moreover the original bilinear system has many solutions and there is some freedom in choosing the coefficient variables and the support variables.
With the choice made in \cite{BBBGNRT19} the information we obtain about the $c_T$'s is not enough to recover 
the coefficient variables
directly. In this case the last step of the algebraic attack still has to compute a Gr\"obner basis for the algebraic system 
consisting of the original system plus the information we have on  the $c_T$'s. 

 Our new approach starts by noticing that there is a better way to use the freedom on the coefficient variables and the support variables: we can actually specify so many coefficient variables that all those that remain unknown  are essentially equal to some maximal minor $c_T$ of $\Cm$. With this we avoid the Gr\"obner basis computation: we obtain from the knowledge of the $c_T$'s 
%obtained from the aforementioned linear system 
the coefficient variables and plugging in theses values in the original bilinear system 
we are left with solving a linear system in the support variables.
      This new approach gives 
%brings on 
a substantial speed-up in the computations for 
      solving the system. It results in the
      best practical efficiency and complexity bounds that are currently known for
      the decoding problem; in particular, it significantly improves upon 
%the aforementioned similar approach in 
\cite{BBBGNRT19}. 
       We present attacks for
        ROLLO-I-128, ROLLO-I-192, and ROLLO-I-256 with bit complexity respectively in
        70, 86, and 158, to be compared to 117, 144, and 197 obtained
        in \cite{BBBGNRT19}. The difference with \cite{BBBGNRT19} is significant since
      as there is no real quantum speed-up for solving linear systems, the best
      quantum attacks for ROLLO-I-192 remained the quantum attack based on combinatorial
      attacks, when our new attacks show that ROLLO parameters are broken and need to be changed.
       
      Our analysis is divided into two categories: the ``overdetermined'' 
      and the ``underdetermined'' case. An $(m,n,k,r)$-decoding instance is overdetermined 
      if 
%the condition 
      \begin{equation}
        \label{condition_intro}
        m\binom{n-k-1}r\ge \binom n r - 1.
       \end{equation}
%      is fulfilled. 
This really corresponds to the case where we have enough linear equations by our approach to find all the $c_T$'s (and hence all the coefficient variables). In that case we obtain a complexity in 
      \begin{equation}
      \mathcal{O}\left(   m\binom{n-p-k-1}r {\binom {n-p} r}^{\omega-1}\right)
      \end{equation}
      operations in the field $\ff q$, where $\omega$ is the constant of
      linear algebra and
      $p = \max\{i : i\in \{1..n\}, m\binom{n-i-k-1}r\ge \binom{n-i}r-1\}$
      represents, in case the overdetermined condition
      \eqref{condition_intro} is comfortably fulfilled, the use of punctured
      codes.  This complexity clearly supersedes the previous results of
      \cite{BBBGNRT19} in terms of complexity and also by the fact that it
      does not require generic 
      Gr\"obner Basis algorithms.  
      In a rough way for $r = \mathcal{O}\left(\sqrt{n}\right)$
      (the type of parameters used for ROLLO and RQC), the recent
      improvements on algebraic attacks can be seen as this: before
      \cite{BBBGNRT19} the complexity for solving RD involved a term in $\mathcal{O}(n^2)$
      in the upper part of a binomial coefficient, the modeling
      in \cite{BBBGNRT19} replaced it by a term in $\mathcal{O}\left(n^{\frac{3}{2}}\right)$ whereas 
      our new modeling involves a term in $\mathcal{O}(n)$ at a similar position. This leads to a gain in the exponential 
      coefficient of order 30 \% compared to
      \cite{BBBGNRT19} and of order 50 \% compared to approaches before
      \cite{BBBGNRT19}. Notice that for ROLLO and RQC only parameters with
      announced complexities 128 and 192 bits satisfied condition (1) but
      not parameters with announced complexities 256 bits.

      When condition (1) is not fulfilled, the instance can 
      either be underdetermined or be brought back to the overdetermined area by an hybrid 
      approach using exhaustive search with exponential complexity to guess few variables
      in the system. In the underdetermined case, our approach is different from 
      \cite{BBBGNRT19}. Here we propose an approach using reduction to the MinRank problem and a new way to solve it. 
Roughly speaking we start with a quadratic modeling of MinRank that we call ``support minors modeling'' which is bilinear in the aforementioned coefficient and support variables and linear in the so called ``linear variables''. The last ones are precisely the $x_i$'s that appear in the MinRank problem. Recall that the coefficient variables are the entries of a $r \times n$ matrix $\Cm$. The crucial observation is now that for all positive integer $b$ all maximal minors of any $(r+b)\times n$ matrix obtained by adding to $\Cm$ any $b$ rows of $\sum_i x_i \Mm_i$ are equal to $0$. These minors are themselves linear combinations of terms of the form
$m c_T$ where $c_T$ is a maximal minor of $\Cm$ and $m$ a monomial of degree $b$ in the $x_i$'s. We can predict the number of independent linear equations in the $m c_T$'s we obtain this way and when the number of such equations is bigger than the number of $m c_T$'s we can recover their values and solve the MinRank problem by linearization. 
This new approach is not only effective in the underdetermined case of the RD problem it can also be quite effective 
for some multivariate proposals made to the NIST competition. In the case of the RD problem, it improves the attacks on \acite{ROLLO} made in \cite{BBBGNRT19} for the parameter sets with the largest values of $r$ (corresponding to 
parameters claiming 256 bits of security). The multivariate schemes that are affected by this new attack are for instance GeMSS and Rainbow. On GeMSS it shows
MinRank attacks together with this new way of solving MinRank come close to the best known attacks against this scheme. On Rainbow it outperforms slightly the best known attacks for certain high security parameter sets. 

At last, not only do these two new ways of solving algebraically the RD or MinRank problem outperform previous algebraic approaches in certain parameter regimes, they are also much better understood: we do not rely on heuristics based on the the first degree fall as in \cite{VBCPS19,BBBGNRT19} to analyze its complexity, but it really amounts to solve a linear system and understand the number of independent linear equations that we obtain which is something for which we have been able to give 
accurate formulas predicting the behavior we obtain experimentally.

      % Note that for some parameters proposed in \acite{ROLLO,RQC2}
      % (in the versions prior to the updates of April 2020), the
      % condition \eqref{condition_intro} holds. Taking for $\omega$ the
      % smallest value currently achievable in practice, which is $\omega\approx2.8$
      % via Strassen's algorithm, this leads to an attack on the schemes proposed in
      % these NIST submissions which is in all cases below the claimed classical
      % security level and sometimes way below the previous attack in \cite{BBBGNRT19}.

      % At last, we propose an analysis explaining 
      % precisely why our attack is significantly more efficient 
      % without the use of generic Gr\"obner basis algorithms.
 
\section{Notation}\label{sec:notation}

      In what follows, we use the following notation and definitions:
      \begin{itemize}
        \item Matrices and vectors are written in boldface font $\mat M$.
        \item The transpose of a matrix \(\mat{M}\) is denoted by \(\trsp{\mat{M}}\).
        \item For a given ring \(\mathcal{R}\), the set of matrices with $n$
          rows, $m$ columns and coefficients in \(\mathcal{R}\) is denoted by
          $\mathcal{R}^{n\times m}$.
        \item $\{1..n\}$ stands for the set of integers from $1$ to $n$.
        \item For a subset $I \subset \{1..n\}$, $\size{I}$ stands for the number of elements in $I$. 
        \item For two subsets $I\subset\{1..n\}$ and $J\subset\{1..m\}$,
          we write $\mat M_{I,J}$ for the submatrix of $\mat M$ formed by its rows
          (resp.~columns) with index in $I$ (resp.~$J$).
        \item For an $m \times n$ matrix $\Mm$ we use the shorthand notation $\mat M_{\any,J} = \mat{M}_{\{1..m\},J}$ and 
% We use the shorthand notations 
%          and $\mat M_{I,\any} = \mat{M}_{I,\{1..n\}}$, where \(\mat M\) has \(m\)
%          rows and \(n\) columns, 
 $\Mm_{i,j}$ for the entry in row $i$ and column $j$.
        \item $\minor{\Mm}{}$ is the determinant of a matrix $\Mm$,
%We denote the determinant of a matrix $\Mm$ by $\minor{\Mm}{}$. We also use a notation inspired by the previous one for denoting the determinant of a submatrix, 
$\minor{\Mm}{I,J}$ is the determinant of the submatrix
          $\Mm_{I,J}$ and $\minor{\Mm}{\any,J}$ is 
%the principal minor of $\Mm$ obtained by taking 
the determinant of
          $\Mm_{\any,J}$.
        \item $\alpha \in \ff{q^m}$ is a primitive element, that is to say that 
          $(1,\alpha,\dots,\alpha^{m-1})$ is a basis of $\ff{q^m}$ seen as an
          $\ff{q}$-vector space.
        \item For $\vec{v}=(v_1,\ldots,v_n) \in \ff{q^m}^n$, the \emph{support} of $\vec v$ is
          the $\ff{q}$-vector subspace of $\ff{q^m}$ spanned by the vectors
          $v_1,\ldots,v_n$. Thus this support is the column space of the matrix
          \(\Mat(\vec{v})\) associated to \(\vec{v}\) (for any choice of basis), and
          its dimension is precisely $\rank{\Mat(\vec{v})}$.
        \item An $[n,k]$ $\Fqm$-linear code is an $\Fqm$-linear subspace of $\Fqm^n$
          of dimension $k$.
%        \item Unless otherwise specified, the \emph{decoding problem} always 
%          refers to the Rank Decoding problem. \bros{Je pense qu'il est utile de rajouter cette remarque, je ne l'ai pas trouvée ailleurs dans l'article et ça évite toute confusion.}
\jp{sincerement je pense que cette remarque n'est pas utile. Cela ne porte vraiment pas a confusion a mon sens.}
       
      \end{itemize}

\section{Algebraic modeling of the MinRank and the decoding problem}\label{sec:modellings}

\subsection{Modeling of MinRank}
The modeling for MinRank  we consider here is related to the modeling used for 
decoding in the rank metric in \cite{BBBGNRT19}. The starting point is that, in order to solve 
\emph{Problem} \ref{problem_minrank0}, 
we look for 
% We start by homogeneizing the MinRank problem by letting $\Mm_{K+1} \eqdef \Ym$, that is we want a nonzero solution $\xv \in \Fq^{K+1}$ such that 
% $$
% \rank{\sum_{i=1}^{K+1}x_i\Mm_i} \leq r.
% $$   
% This is achieved by looking for
a nonzero solution $(\Sm,\Cm,\xv) \in \Fq^{m\times r}\times \Fq^{r\times n}\times \Fq^{K}$ of
\begin{equation}
\label{eq:minrank-homogeneous}
\Sm\Cm = \sum_{i=1}^{K} x_i \Mm_i.
\end{equation}
$\Sm$ is an unknown matrix whose columns give a basis for the column 
space of the matrix $\sum_{i=1}^{K} x_i \Mm_i$ of rank $\leq r$ we are looking for.
The $j$-th column of $\Cm$ represents the coordinates of the $j$-th column 
of the aforementioned matrix in this basis.
We call the entries of $\Sm$ the {\em support variables}, 
and the entries of $\Cm$ the {\em coefficient
variables}. Note that in the above equation, the variables $x_i$ 
only occur linearly. As such, we will dub them the \emph{linear variables}.

Let $\rv_j$ be the $j$-th row of $\sum_{i=1}^{K} x_i \Mm_i$. \eqref{eq:minrank-homogeneous} implies that each row $\rv_j$ is in the rowspace of $\Cm$ (or in coding theoretic terms $\rv_j$ should belong to the 
code $\Cc\eqdef \{\uv \Cm, \uv \in \Fq^r\}$. The following $(r+1)\times n$ matrix $\Cm'_j$ is therefore of rank $\leq r$:
      \[
      \Cm'_j=\begin{pmatrix}
      \vec{r_j} \\
      \Cm
      \end{pmatrix}.
      \]
Therefore, all the maximal minors of this matrix are equal to $0$. These maximal minors 
can be expressed via cofactor expansion with respect to their first row. In this way, they can be seen as bilinear forms in the $x_i$'s and the $r\times r$ minors of $\Cm$. These minors play a fundamental role in the whole paper and  we use the following notation for them.

\begin{notation}
Let $T \subset \{1..n\}$ with $\size{T}=r$. Let $c_T$ be the maximal minor of $\Cm$ corresponding to the columns of $\Cm$ that belong to $T$, i.e.
$$
c_T \eqdef \minor{\Cm}{\any,T}.
$$
\end{notation}

These considerations lead to the following algebraic modeling.
\begin{modelling}[Support Minors modeling]\label{mod:support_minors_modeling}
We consider the system of bilinear equations, given by canceling the maximal minors  of the $m$ matrices $\Cm_j'$:
\begin{equation}\label{eq:MinRank}
\left\{f = 0 \Big| f \in 
      \MaxMinors \begin{pmatrix}
      \vec{r_j} \\
      \Cm
      \end{pmatrix},\;j\in\{1..m\}\right\}.
  \end{equation}
  This system contains:
  \begin{itemize}
\item $m \binom{n}{r+1}$ bilinear equations with coefficients in $\Fq$,
\item  $K+\binom{n}{r}$ unknowns: $\xv=(x_1,\cdots,x_{K})$ and the $c_T$'s, $T \subset \{1..n\}$ with $\size{T}=r$.
\end{itemize}
We search for the solutions $x_i, c_T$'s in $\Fq$.
\end{modelling}

\begin{remark}\mbox{ }
\begin{enumerate}
\item  One of the point of having the $c_T$ as unknowns instead of the coefficients $C_{ij}$ of $\Cm$ is that, if we solve \eqref{eq:MinRank} in the $x_i$ and the $C_{ij}$ variables, then there are many solutions to \eqref{eq:MinRank} since when 
$(\xv,\Cm)$ is a solution for it, then $(\xv,\Am\Cm)$ is also a solution  for any invertible matrix $\Am$ in $\Fq^{r \times r}$.
With the $c_T$ variables we only expect a space of dimension $1$ for the $c_T$ corresponding to the 
transformation $c_T \mapsto \minor{\Am}{} c_T$ that maps a given solution of \eqref{eq:MinRank} to a new one.
\item Another benefit brought by replacing the $C_{ij}$ variables by
  the $c_T$'s is  that it decreases significantly the number
  of possible monomials for writing the algebraic system
  \eqref{eq:MinRank} (about $r!$ times less). This allows 
  for solving this system by linearization when the number of
  equations of the previous modeling exceeds the
  number of different $x_ic_T$ monomials minus 1, namely when
\begin{equation}
\label{eq:linearization}
m\binom{n}{r+1} \geq K\binom{n}{r}-1.
\end{equation}
This turns out  to be ``almost'' the case for several multivariate cryptosystem proposals based on the MinRank problem where $K$ is generally of the same order as $m$ and $n$. 
\end{enumerate}
\end{remark}

\subsection{The  approach followed in \cite{BBBGNRT19} to solve the decoding problem}
      In what follows, we consider the $(m,n,k,r)$-decoding problem 
      for a code $\Cc$ of length $n$, dimension $k$ over $\Fqm$ with a 
      $\yv \in \Fqm^n$ at distance $r$ from $\Cc$ and look for $\cv \in \Cc$ and $\ev$ such that 
      $\yv = \cv + \ev$ and $|\ev|=r$. We assume  that there is a unique solution to this problem (which is relevant for our cryptographic schemes). The starting point is the 
      Ourivksi-Johansson approach, consisting in considering the linear code 
      $\Cy = \mathcal C + \langle \yv\rangle$. From now on, let 
      $\widetilde{G}=(\Im_{k+1} \ \Rm)$ 
      (respectively $\widetilde{H}=(-\trsp{\Rm} \ \Im_{n-k-1})$) 
      be the generator matrix in systematic form (respectively a parity-check matrix) 
      of the extended code $\Cy$.
      By construction, $\ev$ belongs to $\Cy$ as well as all
      its multiples $\lambda \ev$, $\lambda \in \Fqm$. Looking for non-zero codewords in $\Cy$ of rank weight $r$ has at least $q^m-1$ different solutions, namely all the $\lambda \ev$ for 
      $\lambda \in \Fqm^\times$. 
      
      It is readily seen that finding such codewords can be done by solving the (homogeneous) MinRank problem with $\Mm_{ij} \eqdef \Mat(\alpha^{i-1} \cv_j)$ (we adopt a bivariate indexing of the $\Mm_i$'s which is more convenient here), for $(ij) \in \{1..m\}\times
     \{1..k+1\}$ and where $\cv_1,\cdots,\cv_{k+1}$ is an $\Fqm$-basis of $\Cy$.
     This is because the $\alpha^{i-1} \cv_j$'s form an $\Fq$-basis  of $\Cy$. However, the problem with this approach is that $K=(k+1)m=\Th{n^2}$ for the parameters relevant to cryptography. This is much more than for the multivariate cryptosystems based on MinRank and \eqref{eq:linearization} is far from being satisfied here. However, as observed in \cite{BBBGNRT19}, it turns out in this particular case, it is possible because of the $\Fqm$ linear structure of the code, to give an algebraic modeling that only involves the entries of $\Cm$. It is obtained by introducing a parity-check matrix for $\Cy$, that is a matrix $\Hm$ whose kernel is $\Cy$:
     $$
     \Cy = \{\cv \in \Fqm^n:\cv \trsp{\Hm} =0\}.
     $$
In our $\Fqm$ linear setting the solution $\ev$ we are looking for can be written as
\begin{equation}
\label{eq:e}
\ev =  \begin{pmatrix}
                1 & \alpha & \dots & \alpha^{m-1}
              \end{pmatrix}
              \Sm \Cm,
\end{equation}
where $\Sm \in \Fq^{m \times r}$ and $\Cm \in \Fq^{r \times n}$ play the same role as in the previous subsection: $\Sm$ represents a basis of the support of
        $\ev$ in $\left(\Fq^m\right)^r$ and $\Cm$ the coordinates of $\ev$ in
        this basis. By writing that $\ev$ should belong to $\Cy$ we obtain that
        \begin{equation}
        \label{eq:decoding}
        \begin{pmatrix}
                1 & \alpha & \dots & \alpha^{m-1}
              \end{pmatrix}
              \Sm\Cm \trsp{\Hm}=\zerom_{n-k-1}.
        \end{equation} This gives an algebraic system using only the coefficient variables as shown by
            \begin{proposition}[\cite{BBBGNRT19}, Theorem 2]
        The maximal minors of the $r \times (n-k-1)$ matrix $\Cm \trsp{\Hm}$ are all equal to $0$.
        \end{proposition}
        \begin{proof}
        
        Consider the following vector in $\Fqm^r$: $\ev' \eqdef \begin{pmatrix}
                1 & \alpha & \dots & \alpha^{m-1}
              \end{pmatrix}
              \Sm$ whose entries generate (over $\Fq$) the subspace generated by the entries of $\ev$ (i.e. its support). Substituting $\begin{pmatrix}
                1 & \alpha & \dots & \alpha^{m-1}\end{pmatrix}  \Sm$  for $\ev'$ in \eqref{eq:decoding} yields               
                $
              \ev' \Cm \trsp{\Hm}=\zerom_{n-k-1}.
              $
              This shows that the $r \times (n-k-1)$ matrix $\Cm \trsp{\Hm}$ is of rank $\leq r-1$. \qed
              \end{proof}

        These minors $\Cm \trsp{\Hm}$ are polynomials in the entries of $\Cm$ with coefficients in $\Fqm$. Since these entries belong to $\Fq$, the nullity of each minor gives $m$ algebraic equations
        corresponding to polynomials with coefficients in $\Fq$. This involves the following operation
        \begin{notation}
          Let
          ${\mathcal{S}} \eqdef \{\sum_{j} a_{ij} m_{ij}=0,1 \leq i
          \leq N\}$ be a set of polynomial equations where the
          $m_{ij}$'s are the monomials in the unknowns that are
          assumed to belong to $\Fq$, whereas the $a_{ij}$'s are known
          coefficients that belong to $\Fqm$. We define the
          $a_{ijk}$'s as $a_{ij} =\sum_{k=0}^{m-1} a_{ijk} \alpha^k$,
          where the $a_{ijk}$'s belong to $\Fq$. From this we can
          define the system ``unfolding'' over $\Fq$ as
        $$
        \Unfold{(\mathcal{S})} \eqdef \left\{\sum_{j} a_{ijk} m_{ij}=0,1 \leq i \leq N, 0 \leq k \leq m-1\right\}.
        $$
        \end{notation}
The important point is that the solutions of $\mathcal{S}$ over $\Fq$ are
        exactly the solutions of $\Unfold{(\mathcal{S})}$ over $\Fq$,
        so that in that sense the two systems are equivalent.
    
    By using the Cauchy-Binet formula, it is proved \cite[Prop. 1]{BBBGNRT19} that the maximal minors of $\Cm \trsp{\Hm}$,  which are polynomials of degree $\leq r$ in the coefficient variables $C_{ij}$, can actually  be expressed as {\em linear} combinations of the $c_T$'s. In other words we obtain $m \binom{n-k-1}{r}$ linear equations over $\Fq$  by ``unfolding'' the $\binom{n-k-1}{r}$ maximal minors of $\Cm \trsp{\Hm}$.  We denote such a system by
    \begin{equation}
    \label{eq:unfoldmaxminors}
    \Unfold\left(\{f=0 | f \in  \MaxMinors(\Cm\trsp{\Hm})\}\right).
    \end{equation}
    % \bros{In the following paragraph, we need to specify exactly the paragraph mentioned in [11] and give 
    % detail about ``specializing to 1'', for instance ``set the first column to $(1,0,...)^T$}\mb{DONE}
    It is straightforward to check that some variables in $\Cm$ and
    $\Sm$ can be specialized.  The choice which is made
    in~\cite{BBBGNRT19} is to specialize $\Sm$ with its $r$ first rows
    equal to the identity ($\Sm_{\{1..r\},\any}=\Im_r$), its first
    column to $\trsp{\vec 1} = \trsp{(1,0,\dots,0)}$ and $\Cm$ has its first column
    equal to $\trsp{\vec 1}$. It is proved in~\cite[Section
    3.3]{BBBGNRT19} that if the first coordinate of $\ev$ is nonzero
    and the top $r\times r$ block of $\Sm$ is invertible, then the previous specialized system has a unique solution.
 Moreover, this will always be
    the case up to a permutation of the coordinates of the codewords
    or a change of $\Fqm$-basis.

    %In \cite[Prop. 1]{BBBGNRT19} it is proved that $\Sm$ can be specialized with its first column to $\trsp{(1,0,\dots,0)}$ and $\Sm_{\{1..r\},\any}=\Im_r$ and $\Cm$ so that it has its first column equal to $\trsp{(1,0,\dots,0)}$.
    It is proved in     \cite[Prop. 2]{BBBGNRT19} that a degree-$r$ Gr\"obner basis of the unfolded
polynomials  $\MaxMinors$  is obtained by solving the corresponding linear system in the $c_T$'s. 
However, this strategy of specialization does not reveal the entries of $\Cm$  (it only reveals the values of the $c_T$'s). To finish the calculation it still remains to compute a Gr\"obner basis of the whole algebraic system as done in \cite[Step 5, \S 6.1]{BBBGNRT19}). There is a simple way to avoid this computation by specializing the variables of $\Cm$  in a different way. This is the new approach we  explain now.
 
\subsection{The new approach : specializing the identity in $\Cm$}       
As in the previous approach we note that if $(\Sm,\Cm)$ is a
solution of \eqref{eq:decoding} then $(\Sm\Am^{-1},\Am\Cm)$ is also a
solution of it for any invertible matrix $\Am$ in
$\Fq^{r \times r}$. Now, when the first $r$ columns of a
solution $\Cm$ form a invertible matrix, we will still have a solution with the specialization
$$\Cm =   \begin{pmatrix}
  \Im_r & \Cm'
\end{pmatrix}.$$
We can also specialize the first column of $\Sm$ to $\trsp{\vec 1} = \trsp{
  \begin{pmatrix}
    1 & 0 & \dots & 0
  \end{pmatrix}
}$.
If the first $r$ columns of $\Cm$ are not independent, it suffices as in ~\cite[Algo. 1]{BBBGNRT19} to make several different attempts of choosing $r$ columns. The point of this specialization is that 
\begin{itemize}
\item the corresponding $c_T$'s are equal to the entries $C_{ij}$ of
  $\Cm$ up to an unessential factor $(-1)^{r+i}$ whenever
  $T=\{1..r\} \backslash\{i\}\cup\{j\}$ for any $i\in\{1..r\}$
  and $j\in\{r+1..n\}$. This follows on the spot by writing the
  cofactor expansion of the minor
  $c_T = \minor{\Cm}{\any,\{1..r\}\backslash\{i\}\cup\{j\}}$. Solving
  the linear system in the $c_T$'s corresponding to
  \eqref{eq:unfoldmaxminors}
  % \\$\Unfold\left( \MaxMinors(\Cm\trsp{\Hm})\right)$ 
  yields now directly the coefficient variables $C_{ij}$. This avoids the subsequent Gr\"obner basis computation, since once we have $\Cm$ we obtain $\Sm$ directly by solving \eqref{eq:decoding} which has become a linear system.
\item it is readily shown that any solution of \eqref{eq:unfoldmaxminors}
  % \\$\Unfold\left( \MaxMinors(\Cm\trsp{\Hm})\right)$ 
  is actually a projection on the $C_{ij}$ variables of a solution $(\Sm,\Cm)$ of the whole system (see Proposition \ref{prop:projection}). This justifies the whole approach.
\end{itemize}
In other words we are interested here in the following modeling   
\begin{modelling}\label{mod:system} We consider the system of linear equations, given by unfolding all maximal minors of
  $ \begin{pmatrix} \Im_r & \Cm'
        \end{pmatrix}\trsp{\Hm}
        $:
  \begin{equation}
    \label{eq:unfoldmaxminorsnew}
\Unfold\left(    \left\{ f = 0 \Big| f \in  \MaxMinors\left( \begin{pmatrix}
            \Im_r & \Cm'
          \end{pmatrix}\trsp{\Hm}\right)\right\} \right).
  \end{equation}
  This system contains:
\begin{itemize}
\item $m\binom{n-k-1}{r}$ linear equations with coefficients in $\Fq$,
\item   $\binom{n}{r}-1$ unknowns: the $c_T$'s, $T \subset \{1..n\}$ with $\size{T}=r$, 
$T \neq \{1..r\}$.
\end{itemize}
We search for the solutions $c_T$'s in $\Fq$.
\end{modelling}
Note that from the specialization, $c_{\{1..r\}}=1$ is not an unknown.
For the reader's convenience, let us recall the specific form of these equations which is obtained by unfolding the following polynomials (see  \cite[Prop. 2]{BBBGNRT19} and its  proof).
         \begin{proposition}\label{eq:oldarticle}
           $\MaxMinors(\Cm\trsp{\Hm})$ contains
          $\binom{n-k-1}r$ polynomials of degree $r$ over $\ff {q^m}$, indexed
          by the subsets $J\subset\{1..n-k-1\}$ of size $r$, that are the
        \begin{eqnarray}
          \label{eq:PJ}
          P_J &=& 
          \sum_{\substack{T_1\subset\{1..k+1\}, {T_2} \subset J,\\ \size{T_1} + \size{T_2} = r\\T = T_1\cup (T_2+k+1)}}(-1)^{\Sign{T_2}{J}} \minor{\mat R}{T_1, J\backslash T_2} c_T,
        \end{eqnarray} 
        where the sum is over all subsets $T_1\subset\{1..k+1\}$ and
        $T_2$ subset of $J$, with $\size{T_1}+\size{T_2} = r$, and
        $\Sign{T_2}{J}$ is an integer depending on $T_2$ and $J$.  We denote
        by ${T_2}+k+1$ the set $\{i+k+1 : i \in T_2\}$.
        \end{proposition} 

Let us show now that the solutions of this linear system are projections of the solutions of the original system. For this purpose, let us bring in
\begin{itemize}
\item The original system~\eqref{eq:decoding} over $\Fqm$ obtained with the aforementioned specialization
 \begin{eqnarray}
          \label{eqn:oj-modelling-specialized-in-C}
        \mathcal F_C  &=& \left\{
          \begin{pmatrix}
            1 & \alpha & \cdots & \alpha^{m-1}
          \end{pmatrix}
          \begin{pmatrix}
            \trsp{\vec 1} &  \Sm'
          \end{pmatrix}
          \begin{pmatrix}
            \Im_r & \Cm'
          \end{pmatrix}\trsp{\Hm}=\zerom_{n-k-1}
                    \right\},
        \end{eqnarray}
  where $\trsp{\vec 1} = \trsp{
          \begin{pmatrix}
            1 & 0 & \dots & 0
          \end{pmatrix}
      }$, 
        $\Sm =
        \begin{pmatrix}
          \trsp{\vec 1} & \Sm'
        \end{pmatrix}$ 
        and $\Cm =
        \begin{pmatrix}
          \Im_r & \Cm'
        \end{pmatrix}$.        
\item The  system in the coefficient variables we are interested in\\
$
%\begin{equation}
%          \label{eq:maxminors-system}
        \mathcal F_M = \left\{f = 0 \Big| f \in \MaxMinors\left(
          \begin{pmatrix}
        \Im_r & \Cm'
        \end{pmatrix}\trsp{\Hm}
        \right)\right\}.
%        \end{equation}
$
        %which is the set of all minors of size $r$ of the matrix $  \begin{pmatrix}
        %\Im_r & \Cm'
      %\end{pmatrix}\trsp{\Hm}$
\item Let $V_{\ff{q}}({\mathcal F_C})$ be the set of solutions
        of~\eqref{eqn:oj-modelling-specialized-in-C} with all variables in
        $\Fq$, that is\\
 %     \begin{multline}
 $           V_{\ff q}(\mathcal F_C) = $

 $           \left\{(\Sm^*,\Cm^*)\in{\Fq}^{m(r-1) + r(n-r)}  :\begin{pmatrix}
            1 & \alpha & \cdots & \alpha^{m-1}
          \end{pmatrix}
         \begin{pmatrix}
           \trsp{\vec 1} & \Sm^*
        \end{pmatrix}
          \begin{pmatrix}
            \Im_r & \Cm^*
          \end{pmatrix}\trsp{\Hm}
         = \zerov\right\}.$
   %   \end{multline}
      \item
       Let $V_{\ff q}(\mathcal F_M)$ be the set of solutions of 
       $\mathcal F_M$
      % \eqref{eq:maxminors-system}
       with all variables in
        $\ff q$, i.e.\\
        $
         % \begin{eqnarray*}
        V_{\ff q}(\mathcal F_M) = \left\{\Cm^*\in{\ff q}^{r(n-r)}  : \mathtt{Rank}_{\ff {q^m}}\left(\begin{pmatrix}
        \Im_r & \Cm^*
        \end{pmatrix}\trsp{\Hm} \right) < r\right\}.$
     %     \end{eqnarray*}
\end{itemize}
%        Under this assumption, we can specialize
%        System~\eqref{eq:OJmodeling} with the identity in the first columns of
%        $\Cm$, and the value $\trsp{\vec 1} = \trsp{
%          \begin{pmatrix}
%            1 & 0 & \dots & 0
%          \end{pmatrix}
%      }$ in the first column of $\Sm$. Precisely, we define 
%        \begin{eqnarray}
%          \label{eqn:oj-modelling-specialized-in-C}
%        \mathcal F_C  &=& \left\{
%          \begin{pmatrix}
%            1 & \alpha & \cdots & \alpha^{m-1}
%          \end{pmatrix}
%          \begin{pmatrix}
%            \trsp{\vec 1} &  \Sm'
%          \end{pmatrix}
%          \begin{pmatrix}
%            \Im_r & \Cm'
%          \end{pmatrix}\trsp{\Hm}
%                    \right\},
%        \end{eqnarray}
%        where $\trsp{\vec 1}\in \Fq^m$ is a column vector,
%        $\Sm =
%        \begin{pmatrix}
%          \trsp{\vec 1} & \Sm'
%        \end{pmatrix}$ 
%        and $\Cm =
%        \begin{pmatrix}
%          \Im_r & \Cm'
%        \end{pmatrix}$.

        With these notations at hand, we  now show that solving the decoding problem is left to solve the
        $\MaxMinors$ system depending only on the $\Cm$ variables.
        %        \begin{equation}
%          \label{eq:maxminors-system}
%        \mathcal F_M = \MaxMinors\left(
%          \begin{pmatrix}
%        \Im_r & \Cm'
%        \end{pmatrix}\trsp{\Hm}
%        \right),
%        \end{equation}
%        which is the set of all minors of size $r$ of the matrix $  \begin{pmatrix}
%        \Im_r & \Cm'
%      \end{pmatrix}\trsp{\Hm}$, can be used to recover the values of the variables in $\Cm$.
%
%        Let $V_{\ff{q}}({\mathcal F_C})$ be the set of solutions
%        of~\eqref{eqn:oj-modelling-specialized-in-C} with all variables in
%        $\Fq$, that is
%      \begin{multline}
%            V_{\ff q}(\mathcal F_C) = \\
%            \left\{(\Sm^*,\Cm^*)\in{\Fq}^{m(r-1) + r(n-r)}  :\begin{pmatrix}
%            1 & \alpha & \cdots & \alpha^{m-1}
%          \end{pmatrix}
%         \begin{pmatrix}
%           \trsp{\vec 1} & \Sm^*
%        \end{pmatrix}
%          \begin{pmatrix}
%            \Im_r & \Cm^*
%          \end{pmatrix}\trsp{\Hm}
%         = \zerov\right\}.
%      \end{multline}
%       Let $V_{\ff q}(\mathcal F_M)$ be the set of solutions of
%       \eqref{eq:maxminors-system} with all variables in
%        $\ff q$, i.e.
%          \begin{eqnarray*}
%        V_{\ff q}(\mathcal F_M) &=& \left\{\Cm^*\in{\ff q}^{r(n-r)}  : \mathtt{Rank}_{\ff {q^m}}\left(\begin{pmatrix}
%        \Im_r & \Cm^*
%        \end{pmatrix}\trsp{\Hm} \right) < r\right\}.
%          \end{eqnarray*}

          \begin{proposition}\label{prop:projection}
            If $\ev$ can be uniquely decoded and has rank $r$, then 
            \begin{equation}\label{eq:solutionsmaxminors}
            V_{\ff q}(\mathcal F_M) = \left\{\cm^*\in\Fq^{r(n-r)} : \exists \Sm^*\in\Fq^{m(r-1)} \text{ s.t. } (\Sm^*,\cm^*)\in V_{\ff q}(\mathcal F_C)\right\},    
          \end{equation}
          that is $V_{\ff q}(\mathcal F_M)$ is the projection
          of  $V_{\ff q}(\mathcal F_C)$ on the last $r(n-r)$ coordinates.
        \end{proposition}

        \begin{proof}
          Let $(\Sm^*, \cm^*)\in V_{\ff q}(\mathcal F_C)$, then 
                 % \begin{equation*}
       $ \begin{pmatrix}
            1 & S_2^* & \dots S_r^*
          \end{pmatrix} =   \begin{pmatrix}
            1 & \alpha & \cdots & \alpha^{m-1}
          \end{pmatrix}
          \begin{pmatrix}
            \trsp{\vec 1} & \Sm^*
          \end{pmatrix}$
       % \end{equation*} 
        belongs to the left kernel of the matrix $
          \begin{pmatrix}
        \Im_r & \cm^*
        \end{pmatrix}\trsp{\Hm}$. Hence this matrix has rank less than $r$,
        and $\cm^*\in V_{\ff q}(\mathcal F_M)$. Reciprocally, if
        $\cm^*\in V_{\ff q}(\mathcal F_M)$, then the matrix
        $\begin{pmatrix} \Im_r & \cm^*
        \end{pmatrix}\trsp{\Hm}$ has rank less than $r$, hence its left kernel over
        $\ff {q^m}$ contains a non zero element
        $(S_1^*,\dots,S_r^*) =(1,\alpha,\dots,\alpha^{m-1})\Sm^*$ with the
        coefficients of $\Sm^*$ in $\ff q$. But $S_1^*$ cannot be zero, as it
        would mean that $(0,S_2^*,\dots,S_r^*)
        \begin{pmatrix}
          \Im_r&\cm^*
        \end{pmatrix}$ is an error of weight less than $r$ solution of the
        decoding problem, and we assumed there is only one error of weight
        exactly $r$ solution of the decoding problem. Then,
        $({S_1^*}^{-1}(S_2^*,\dots,S_r^*),\cm^*)\in V_{\ff q}(\mathcal F_C)$.
        \qed
        \end{proof}
%        This means that solving the decoding problem is left to solve the
%        $\MaxMinors$ system, that depends only on the $\Cm$ variables.
%        \begin{proposition} The system $\MaxMinors(\Cm\trsp{\Hm})$ contains
%          $\binom{n-k-1}r$ polynomials of degree $r$ over $\ff {q^m}$, indexed
%          by the subsets $J\subset\{1..n-k-1\}$ of size $r$, that are the
%        \begin{eqnarray}
%          \label{eq:PJ}
%          P_J &=& 
%          \sum_{\substack{\mat {T_1}\subset\{1..k+1\}, {\mat {T_2}}\subset J,\\ \size{\mat {T_1}} + \size{\mat {T_2}} = r\\{\mat T} = {\mat {T_1}} \cup ({\mat {T_2}}+k+1)}}(-1)^{\Sign{T_2}{J}} \det(\mat R_{{\mat {T_1}}, J\backslash{\mat {T_2}}}) \det(\Cm_{\any,{\mat T}}), 
%        \end{eqnarray} 
%        where the sum is over all subsets $\mat T_1\subset\{1..k+1\}$ and
%        $\mat T_2$ subset of $J$, with $\size{\mat T_1}+\size{\mat T_2} = r$, and
%        $\Sign{T_2}{J}$ is an integer depending on $T_2$ and $J$.  We denote
%        by ${\mat {T_2}}+k+1$ the set $\{i+k+1 : i \in {\mat {T_2}}\}$.
%        \end{proposition}
%        \begin{remark}\label{rem:ct}
%          There are $\binom{n}r$ different polynomials
%          $\det(\Cm_{\any,{\mat T}})$ involved in the $\binom{n-k-1}r$
%          equations, and each equation $P_J$ contains $\binom{k+r+1}{r}$ such
%          polynomials.
%
%        We have $c_{i,j} = \det(\Cm_{\any,\{1..r\}\backslash\{i\}\cup\{j\}})$
%        for any $i\in\{1..r\}$ and $j\in\{r+1..n\}$, and
%        $1 = \det(\Cm_{\any,\{1..r\}})$.
%        \end{remark}
%        For the proof, reader may refer to \cite{BBBGNRT19}.

\section{Solving RD: overdetermined case\label{attack_1}}

        \pgfplotstableread[row sep=\\]{
        X Y \\20 60.88 \\ 22 64.07 \\ 24 66.98 \\ 26 69.30 \\ 28 71.38 \\ 30 69.19 \\ 32 70.86 \\ 34 72.41 \\ 36 73.94 \\ 38 70.95 \\ 40 72.18 \\ 42 73.34 \\ 44 69.78 \\ 46 70.95 \\ 48 67.23 \\ 50 68.22 \\ 52 64.49 \\ 54 65.32 \\ 56 66.12 \\ 58 61.77 \\ 60 62.52 \\ 62 63.25 \\ 64 63.62 \\ 66 64.31 \\ 68 64.52 \\ 70 65.18 \\ 72 65.87 \\ 74 65.76 \\ 76 66.39 \\ 78 67.00 \\ 80 67.59 \\ 82 67.85 \\ 84 68.42 \\ 86 68.63 \\ 88 68.91 \\ 90 69.08 \\ 92 69.62 \\ 94 70.15 \\ 96 70.38 \\ 98 70.89 \\ 100 71.14 \\ 102 71.29 \\ 104 71.78 \\ 106 72.26 \\ 108 72.19 \\ 110 72.66 \\ 112 73.12 \\ 114 73.26 \\ 116 73.45 \\ 118 73.90 \\ 120 74.33 \\ 122 74.49 \\ 124 74.67 \\ 126 75.09 \\ 128 75.24 \\ 130 75.41 \\ 132 75.82 \\ 134 76.22 \\ 136 76.23 \\ 138 76.34 \\ 140 76.74 \\ 142 77.12 \\ 144 77.24 \\ 146 77.61 \\ 148 77.72 \\ 150 78.09 \\ 152 78.24 \\ 154 78.60 \\ 156 78.77 \\ 158 78.87 \\ 160 79.22 \\ 162 79.32 \\ 164 79.69 \\ 166 79.64 \\ 168 79.98 \\ 170 80.32 \\ 172 80.41 \\ 174 80.74 \\ 176 80.83 \\ 178 81.17 \\ 180 81.32 \\ 182 81.41 \\ 184 81.72 \\ 186 81.81 \\ 188 81.95 \\ 190 82.26 \\ 192 82.34 \\ 194 82.47 \\ 196 82.76 \\ 198 83.06 \\ 200 83.19 \\ 202 83.28 \\ 204 83.56 \\ 206 83.69 \\ 208 83.77 \\ 210 84.06 \\ 212 84.14 \\ 214 84.43 \\ 216 84.39 \\ 218 84.66 \\ 220 84.74 \\ 222 85.01 \\ 224 85.08 \\ 226 85.35 \\ 228 85.44 \\ 230 85.54 \\ 232 85.80 \\ 234 85.89 \\ 236 86.05 \\ 238 86.12 \\ 240 86.37 \\ 242 86.44 \\ 244 86.69 \\ 246 86.76 \\ 248 87.01 \\ 250 87.08 \\ 252 87.23 \\ 254 87.29 \\ 256 87.54 \\ 258 87.60 \\ 260 87.84 \\ 262 87.91 \\ 264 88.14 \\ 266 88.23 \\ 268 88.30 \\ 270 88.54 \\ 272 88.63 \\ 274 88.69 \\ 276 88.92 \\ 278 89.02 \\ 280 89.08 \\ 282 89.30 \\ 284 89.37 \\ 286 89.49 \\ 288 89.55 \\ 290 89.77 \\ 292 89.83 \\ 294 90.05 \\ 296 90.11 \\ 298 90.20 \\ 300 90.41 \\ 
        }{\dataTheoOptCinq}

        \pgfplotstableread[row sep=\\]{
          X Y\\20 72.88 \\ 22 80.83 \\ 24 84.80 \\ 26 87.99 \\ 28 90.81 \\ 30 93.66 \\ 32 95.93 \\ 34 98.01 \\ 36 100.0 \\ 38 102.0 \\ 40 103.6 \\ 42 105.2 \\ 44 106.6 \\ 46 108.2 \\ 48 109.5 \\ 50 110.8 \\ 52 106.4 \\ 54 107.5 \\ 56 113.5 \\ 58 109.8 \\ 60 110.7 \\ 62 111.7 \\ 64 107.0 \\ 66 107.8 \\ 68 108.7 \\ 70 109.5 \\ 72 110.3 \\ 74 105.4 \\ 76 106.1 \\ 78 106.8 \\ 80 107.5 \\ 82 102.5 \\ 84 103.1 \\ 86 97.99 \\ 88 92.90 \\ 90 93.48 \\ 92 94.05 \\ 94 94.60 \\ 96 89.38 \\ 98 89.91 \\ 100 84.69 \\ 102 85.19 \\ 104 85.68 \\ 106 86.17 \\ 108 86.11 \\ 110 86.59 \\ 112 87.06 \\ 114 87.52 \\ 116 87.72 \\ 118 88.16 \\ 120 88.60 \\ 122 88.77 \\ 124 88.96 \\ 126 89.38 \\ 128 89.83 \\ 130 90.01 \\ 132 90.42 \\ 134 90.82 \\ 136 90.84 \\ 138 90.97 \\ 140 91.36 \\ 142 91.75 \\ 144 92.14 \\ 146 92.52 \\ 148 92.63 \\ 150 93.00 \\ 152 93.16 \\ 154 93.53 \\ 156 93.70 \\ 158 94.06 \\ 160 94.41 \\ 162 94.76 \\ 164 94.88 \\ 166 95.09 \\ 168 95.19 \\ 170 95.53 \\ 172 95.86 \\ 174 96.19 \\ 176 96.29 \\ 178 96.64 \\ 180 96.79 \\ 182 97.11 \\ 184 97.20 \\ 186 97.51 \\ 188 97.66 \\ 190 97.97 \\ 192 98.27 \\ 194 98.40 \\ 196 98.70 \\ 198 98.79 \\ 200 98.93 \\ 202 99.22 \\ 204 99.51 \\ 206 99.65 \\ 208 99.73 \\ 210 100.0 \\ 212 100.3 \\ 214 100.4 \\ 216 100.5 \\ 218 100.8 \\ 220 100.9 \\ 222 101.1 \\ 224 101.4 \\ 226 101.5 \\ 228 101.8 \\ 230 101.9 \\ 232 102.1 \\ 234 102.2 \\ 236 102.4 \\ 238 102.5 \\ 240 102.7 \\ 242 103.0 \\ 244 103.1 \\ 246 103.3 \\ 248 103.6 \\ 250 103.8 \\ 252 103.8 \\ 254 104.0 \\ 256 104.1 \\ 258 104.3 \\ 260 104.6 \\ 262 104.6 \\ 264 104.9 \\ 266 105.0 \\ 268 105.2 \\ 270 105.3 \\ 272 105.5 \\ 274 105.6 \\ 276 105.8 \\ 278 105.9 \\ 280 106.2 \\ 282 106.2 \\ 284 106.5 \\ 286 106.4 \\ 288 106.7 \\ 290 106.9 \\ 292 106.9 \\ 294 107.2 \\ 296 107.4 \\ 298 107.5 \\ 300 107.5 \\ 
         }{\dataTheoOptSix}
          
        \pgfplotstableread[row sep=\\]{
        X Y\\ 20 82.56 \\ 22 93.57 \\ 24 98.75 \\ 26 105.7 \\ 28 112.2 \\ 30 115.9 \\ 32 118.9 \\ 34 125.9 \\ 36 128.6 \\ 38 131.3 \\ 40 135.0 \\ 42 136.9 \\ 44 138.8 \\ 46 140.7 \\ 48 147.7 \\ 50 149.4 \\ 52 145.5 \\ 54 152.5 \\ 56 154.8 \\ 58 155.4 \\ 60 157.5 \\ 62 158.7 \\ 64 153.3 \\ 66 160.3 \\ 68 162.1 \\ 70 163.1 \\ 72 164.1 \\ 74 158.5 \\ 76 165.5 \\ 78 167.0 \\ 80 167.8 \\ 82 168.2 \\ 84 169.6 \\ 86 169.9 \\ 88 164.5 \\ 90 165.2 \\ 92 166.0 \\ 94 166.7 \\ 96 167.4 \\ 98 168.1 \\ 100 162.0 \\ 102 162.7 \\ 104 163.3 \\ 106 170.3 \\ 108 157.8 \\ 110 164.8 \\ 112 165.8 \\ 114 166.3 \\ 116 160.1 \\ 118 160.6 \\ 120 161.2 \\ 122 161.7 \\ 124 155.4 \\ 126 162.4 \\ 128 156.4 \\ 130 157.0 \\ 132 157.4 \\ 134 157.9 \\ 136 144.8 \\ 138 145.2 \\ 140 145.7 \\ 142 146.1 \\ 144 146.5 \\ 146 147.0 \\ 148 147.4 \\ 150 147.8 \\ 152 141.3 \\ 154 141.7 \\ 156 135.3 \\ 158 135.7 \\ 160 136.1 \\ 162 136.4 \\ 164 136.8 \\ 166 123.5 \\ 168 123.8 \\ 170 124.2 \\ 172 124.5 \\ 174 124.9 \\ 176 125.2 \\ 178 118.7 \\ 180 112.2 \\ 182 112.5 \\ 184 112.8 \\ 186 113.1 \\ 188 113.3 \\ 190 113.6 \\ 192 113.9 \\ 194 114.0 \\ 196 114.3 \\ 198 114.6 \\ 200 114.8 \\ 202 115.1 \\ 204 115.4 \\ 206 115.5 \\ 208 115.8 \\ 210 116.1 \\ 212 116.3 \\ 214 116.6 \\ 216 116.6 \\ 218 116.9 \\ 220 117.2 \\ 222 117.4 \\ 224 117.7 \\ 226 118.0 \\ 228 118.1 \\ 230 118.2 \\ 232 118.5 \\ 234 118.7 \\ 236 118.7 \\ 238 119.0 \\ 240 119.2 \\ 242 119.5 \\ 244 119.7 \\ 246 120.0 \\ 248 120.1 \\ 250 120.3 \\ 252 120.5 \\ 254 120.6 \\ 256 120.8 \\ 258 121.0 \\ 260 121.3 \\ 262 121.5 \\ 264 121.8 \\ 266 121.9 \\ 268 122.1 \\ 270 122.2 \\ 272 122.4 \\ 274 122.7 \\ 276 122.7 \\ 278 123.0 \\ 280 123.1 \\ 282 123.3 \\ 284 123.5 \\ 286 123.5 \\ 288 123.7 \\ 290 123.9 \\ 292 124.2 \\ 294 124.4 \\ 296 124.4 \\ 298 124.7 \\ 300 124.8 \\ 
         }{\dataTheoOptSept}

        \pgfplotstableread[row sep=\\]{
        X Y\\ 22 100.7 \\ 24 112.7 \\ 26 120.7 \\ 28 129.1 \\ 30 133.6 \\ 32 141.6 \\ 34 147.5 \\ 36 150.7 \\ 38 153.8 \\ 40 161.8 \\ 42 165.9 \\ 44 173.9 \\ 46 176.3 \\ 48 179.8 \\ 50 181.8 \\ 52 183.7 \\ 54 191.7 \\ 56 194.5 \\ 58 196.2 \\ 60 197.7 \\ 62 205.7 \\ 64 200.7 \\ 66 208.7 \\ 68 210.8 \\ 70 212.1 \\ 72 213.3 \\ 74 214.6 \\ 76 215.7 \\ 78 223.7 \\ 80 225.4 \\ 82 226.5 \\ 84 227.5 \\ 86 228.5 \\ 88 229.6 \\ 90 230.5 \\ 92 231.3 \\ 94 239.3 \\ 96 233.1 \\ 98 241.1 \\ 100 234.8 \\ 102 242.8 \\ 104 244.0 \\ 106 244.8 \\ 108 237.9 \\ 110 245.9 \\ 112 247.0 \\ 114 247.7 \\ 116 248.4 \\ 118 249.1 \\ 120 249.7 \\ 122 250.4 \\ 124 251.1 \\ 126 251.7 \\ 128 252.3 \\ 130 252.9 \\ 132 253.5 \\ 134 254.0 \\ 136 246.9 \\ 138 247.5 \\ 140 248.0 \\ 142 248.5 \\ 144 249.0 \\ 146 257.0 \\ 148 257.9 \\ 150 258.4 \\ 152 258.6 \\ 154 259.4 \\ 156 252.1 \\ 158 252.5 \\ 160 253.0 \\ 162 261.0 \\ 164 253.9 \\ 166 246.6 \\ 168 247.0 \\ 170 255.0 \\ 172 255.7 \\ 174 256.1 \\ 176 256.5 \\ 178 257.0 \\ 180 249.5 \\ 182 249.9 \\ 184 250.3 \\ 186 258.3 \\ 188 251.1 \\ 190 251.5 \\ 192 251.9 \\ 194 252.3 \\ 196 252.7 \\ 198 253.0 \\ 200 245.5 \\ 202 245.9 \\ 204 246.2 \\ 206 246.6 \\ 208 247.0 \\ 210 247.3 \\ 212 247.6 \\ 214 248.0 \\ 216 232.6 \\ 218 240.6 \\ 220 241.1 \\ 222 241.4 \\ 224 241.7 \\ 226 242.0 \\ 228 242.4 \\ 230 234.8 \\ 232 235.1 \\ 234 235.4 \\ 236 220.0 \\ 238 228.0 \\ 240 228.4 \\ 242 228.7 \\ 244 229.0 \\ 246 229.3 \\ 248 229.6 \\ 250 229.8 \\ 252 214.4 \\ 254 214.6 \\ 256 214.9 \\ 258 215.2 \\ 260 215.4 \\ 262 215.7 \\ 264 216.0 \\ 266 216.2 \\ 268 216.5 \\ 270 216.8 \\ 272 209.1 \\ 274 209.4 \\ 276 209.6 \\ 278 202.0 \\ 280 202.2 \\ 282 202.4 \\ 284 202.7 \\ 286 195.1 \\ 288 195.3 \\ 290 195.5 \\ 292 195.8 \\ 294 196.0 \\ 296 196.2 \\ 298 188.5 \\ 300 188.8 \\ 
         }{\dataTheoOptHuit}

        \pgfplotstableread[row sep=\\]{
        X Y\\ 24 120.9 \\ 26 133.5 \\ 28 142.5 \\ 30 151.9 \\ 32 160.9 \\ 34 168.1 \\ 36 177.2 \\ 38 181.2 \\ 40 186.6 \\ 42 195.6 \\ 44 200.3 \\ 46 203.1 \\ 48 212.1 \\ 50 216.0 \\ 52 218.3 \\ 54 227.3 \\ 56 230.6 \\ 58 232.7 \\ 60 241.7 \\ 62 244.6 \\ 64 246.4 \\ 66 255.4 \\ 68 257.9 \\ 70 259.4 \\ 72 268.5 \\ 74 270.0 \\ 76 272.2 \\ 78 273.5 \\ 80 282.5 \\ 82 283.8 \\ 84 285.8 \\ 86 286.9 \\ 88 288.1 \\ 90 297.1 \\ 92 298.9 \\ 94 299.9 \\ 96 301.0 \\ 98 310.0 \\ 100 303.0 \\ 102 312.0 \\ 104 313.5 \\ 106 322.5 \\ 108 315.3 \\ 110 316.2 \\ 112 325.2 \\ 114 326.5 \\ 116 327.3 \\ 118 328.1 \\ 120 337.1 \\ 122 338.4 \\ 124 338.7 \\ 126 339.9 \\ 128 340.6 \\ 130 341.4 \\ 132 342.0 \\ 134 351.0 \\ 136 343.1 \\ 138 344.1 \\ 140 344.8 \\ 142 345.4 \\ 144 354.4 \\ 146 355.4 \\ 148 356.0 \\ 150 365.0 \\ 152 365.6 \\ 154 366.6 \\ 156 358.4 \\ 158 367.4 \\ 160 368.3 \\ 162 368.8 \\ 164 369.4 \\ 166 369.6 \\ 168 370.5 \\ 170 371.0 \\ 172 371.5 \\ 174 380.5 \\ 176 381.3 \\ 178 381.8 \\ 180 373.5 \\ 182 382.5 \\ 184 383.3 \\ 186 383.7 \\ 188 384.2 \\ 190 384.7 \\ 192 385.1 \\ 194 385.6 \\ 196 386.0 \\ 198 395.0 \\ 200 386.9 \\ 202 387.3 \\ 204 396.3 \\ 206 388.2 \\ 208 388.6 \\ 210 397.6 \\ 212 398.2 \\ 214 398.6 \\ 216 390.2 \\ 218 390.6 \\ 220 391.0 \\ 222 400.0 \\ 224 400.6 \\ 226 400.9 \\ 228 401.3 \\ 230 401.5 \\ 232 402.1 \\ 234 402.4 \\ 236 394.0 \\ 238 394.3 \\ 240 394.7 \\ 242 395.0 \\ 244 395.3 \\ 246 404.3 \\ 248 404.9 \\ 250 405.2 \\ 252 396.7 \\ 254 397.0 \\ 256 397.4 \\ 258 397.7 \\ 260 406.7 \\ 262 407.2 \\ 264 407.5 \\ 266 407.6 \\ 268 408.1 \\ 270 408.4 \\ 272 399.9 \\ 274 408.9 \\ 276 409.3 \\ 278 400.8 \\ 280 401.1 \\ 282 410.1 \\ 284 410.4 \\ 286 401.8 \\ 288 402.2 \\ 290 402.5 \\ 292 402.8 \\ 294 403.1 \\ 296 403.3 \\ 298 403.6 \\ 300 403.9 \\ 
         }{\dataTheoOptNeuf}

        \pgfplotstableread[row sep=\\]{
        X Y\\ 21 40.28 \\ 24 41.38 \\ 27 43.27 \\ 30 44.19 \\ 33 45.86 \\ 36 47.50 \\ 39 48.23 \\ 42 49.64 \\ 45 50.21 \\ 48 51.50 \\ 51 52.78 \\ 54 53.27 \\ 57 53.81 \\ 60 54.91 \\ 63 55.95 \\ 66 56.42 \\ 69 57.40 \\ 72 57.74 \\ 75 58.16 \\ 78 59.07 \\ 81 59.43 \\ 84 60.28 \\ 87 61.15 \\ 90 61.51 \\ 93 61.75 \\ 96 62.59 \\ 99 63.34 \\ 102 63.66 \\ 105 63.87 \\ 108 64.21 \\ 111 64.91 \\ 114 65.58 \\ 117 65.83 \\ 120 66.48 \\ 123 66.75 \\ 126 66.94 \\ 129 67.58 \\ 132 67.81 \\ 135 68.15 \\ 138 68.33 \\ 141 68.91 \\ 144 69.48 \\ 147 69.64 \\ 150 70.19 \\ 153 70.39 \\ 156 70.61 \\ 159 71.13 \\ 162 71.28 \\ 165 71.81 \\ 168 72.05 \\ 171 72.20 \\ 174 72.69 \\ 177 72.84 \\ 180 73.04 \\ 183 73.51 \\ 186 73.97 \\ 189 74.15 \\ 192 74.28 \\ 195 74.76 \\ 198 74.89 \\ 201 75.06 \\ 204 75.49 \\ 207 75.66 \\ 210 76.08 \\ 213 76.22 \\ 216 76.41 \\ 219 76.52 \\ 222 76.93 \\ 225 77.04 \\ 228 77.45 \\ 231 77.59 \\ 234 77.71 \\ 237 77.90 \\ 240 78.29 \\ 243 78.39 \\ 246 78.76 \\ 249 78.87 \\ 252 79.05 \\ 255 79.41 \\ 258 79.51 \\ 261 79.86 \\ 264 79.96 \\ 267 80.08 \\ 270 80.44 \\ 273 80.56 \\ 276 80.90 \\ 279 81.03 \\ 282 81.12 \\ 285 81.46 \\ 288 81.61 \\ 291 81.70 \\ 294 82.02 \\ 297 82.11 \\ 300 82.23 \\ 303 82.54 \\ 306 82.66 \\ 309 82.74 \\ 312 82.86 \\ 315 83.17 \\ 318 83.25 \\ 321 83.37 \\ 324 83.67 \\ 327 83.75 \\ 330 84.07 \\ 333 84.16 \\ 336 84.29 \\ 339 84.36 \\ 342 84.65 \\ 345 84.73 \\ 348 84.87 \\ 351 85.15 \\ 354 85.23 \\ 357 85.50 \\ 360 85.57 \\ 363 85.65 \\ 366 85.94 \\ 369 86.02 \\ 372 86.11 \\ 375 86.37 \\ 378 86.31 \\ 381 86.58 \\ 384 86.64 \\ 387 86.90 \\ 390 86.97 \\ 393 87.07 \\ 396 87.32 \\ 399 87.43 \\ 402 87.50 \\ 405 87.75 \\ 408 87.81 \\ 411 88.07 \\ 414 88.15 \\ 417 88.21 \\ 420 88.30 \\ 423 88.54 \\ 426 88.64 \\ 429 88.70 \\ 432 88.94 \\ 435 89.02 \\ 438 89.08 \\ 441 89.31 \\ 444 89.40 \\ 447 89.46 \\ 450 89.70 \\ 453 89.76 \\ 456 89.84 \\ 459 90.07 \\ 462 90.16 \\ 465 90.22 \\ 468 90.44 \\ 471 90.51 \\ 474 90.57 \\ 477 90.65 \\ 480 90.87 \\ 483 90.92 \\ 486 91.02 \\ 489 91.07 \\ 492 91.28 \\ 495 91.34 \\ 498 91.44 \\ 
         }{\dataTheoOptCinqRQC}

        \pgfplotstableread[row sep=\\]{
        X Y\\21 53.44 \\ 24 52.85 \\ 27 54.99 \\ 30 53.11 \\ 33 54.80 \\ 36 56.47 \\ 39 57.24 \\ 42 58.69 \\ 45 60.20 \\ 48 61.50 \\ 51 62.80 \\ 54 63.32 \\ 57 63.87 \\ 60 65.01 \\ 63 66.79 \\ 66 67.28 \\ 69 68.27 \\ 72 68.63 \\ 75 69.71 \\ 78 70.62 \\ 81 70.99 \\ 84 71.85 \\ 87 72.73 \\ 90 73.10 \\ 93 73.90 \\ 96 74.22 \\ 99 74.99 \\ 102 75.83 \\ 105 76.06 \\ 108 76.41 \\ 111 77.12 \\ 114 77.80 \\ 117 78.53 \\ 120 78.74 \\ 123 79.46 \\ 126 80.09 \\ 129 80.31 \\ 132 80.97 \\ 135 80.91 \\ 138 81.51 \\ 141 82.10 \\ 144 82.67 \\ 147 82.84 \\ 150 83.39 \\ 153 83.98 \\ 156 84.21 \\ 159 84.73 \\ 162 84.89 \\ 165 85.43 \\ 168 85.68 \\ 171 86.18 \\ 174 86.33 \\ 177 86.83 \\ 180 87.04 \\ 183 87.51 \\ 186 87.97 \\ 189 88.17 \\ 192 88.62 \\ 195 88.79 \\ 198 89.23 \\ 201 89.42 \\ 204 89.85 \\ 207 90.03 \\ 210 90.45 \\ 213 90.59 \\ 216 90.79 \\ 219 91.20 \\ 222 91.61 \\ 225 91.73 \\ 228 92.14 \\ 231 92.29 \\ 234 92.69 \\ 237 92.89 \\ 240 93.00 \\ 243 93.38 \\ 246 93.75 \\ 249 93.86 \\ 252 94.05 \\ 255 94.41 \\ 258 94.77 \\ 261 94.88 \\ 264 95.23 \\ 267 95.36 \\ 270 95.72 \\ 273 95.84 \\ 276 96.18 \\ 279 96.32 \\ 282 96.65 \\ 285 96.76 \\ 288 96.92 \\ 291 97.24 \\ 294 97.56 \\ 297 97.66 \\ 300 97.79 \\ 303 98.10 \\ 306 98.23 \\ 309 98.54 \\ 312 98.66 \\ 315 98.96 \\ 318 99.05 \\ 321 99.39 \\ 324 99.48 \\ 327 99.78 \\ 330 99.88 \\ 333 100.1 \\ 336 100.3 \\ 339 100.4 \\ 342 100.6 \\ 345 100.9 \\ 348 100.9 \\ 351 101.2 \\ 354 101.4 \\ 357 101.5 \\ 360 101.8 \\ 363 101.9 \\ 366 102.2 \\ 369 102.2 \\ 372 102.5 \\ 375 102.6 \\ 378 102.7 \\ 381 103.0 \\ 384 103.1 \\ 387 103.3 \\ 390 103.4 \\ 393 103.7 \\ 396 103.8 \\ 399 103.9 \\ 402 104.1 \\ 405 104.2 \\ 408 104.4 \\ 411 104.5 \\ 414 104.8 \\ 417 104.9 \\ 420 105.1 \\ 423 105.2 \\ 426 105.3 \\ 429 105.5 \\ 432 105.8 \\ 435 105.9 \\ 438 105.9 \\ 441 106.1 \\ 444 106.2 \\ 447 106.5 \\ 450 106.5 \\ 453 106.8 \\ 456 106.9 \\ 459 107.1 \\ 462 107.2 \\ 465 107.2 \\ 468 107.5 \\ 471 107.5 \\ 474 107.8 \\ 477 107.8 \\ 480 107.9 \\ 483 108.1 \\ 486 108.2 \\ 489 108.4 \\ 492 108.5 \\ 495 108.7 \\ 498 108.8 \\ 
         }{\dataTheoOptSixRQC}

        \pgfplotstableread[row sep=\\]{
        X Y\\ 21 66.07 \\ 24 66.61 \\ 27 71.35 \\ 30 69.52 \\ 33 67.86 \\ 36 70.84 \\ 39 66.23 \\ 42 68.66 \\ 45 70.18 \\ 48 71.49 \\ 51 72.81 \\ 54 74.13 \\ 57 74.70 \\ 60 75.83 \\ 63 76.92 \\ 66 78.12 \\ 69 79.13 \\ 72 80.15 \\ 75 80.60 \\ 78 81.54 \\ 81 82.52 \\ 84 83.39 \\ 87 84.28 \\ 90 84.67 \\ 93 86.02 \\ 96 86.35 \\ 99 87.12 \\ 102 87.97 \\ 105 88.70 \\ 108 89.06 \\ 111 89.77 \\ 114 90.46 \\ 117 90.73 \\ 120 91.40 \\ 123 92.14 \\ 126 92.77 \\ 129 93.42 \\ 132 93.68 \\ 135 94.04 \\ 138 94.64 \\ 141 95.23 \\ 144 95.81 \\ 147 96.38 \\ 150 96.93 \\ 153 97.15 \\ 156 97.76 \\ 159 98.29 \\ 162 98.81 \\ 165 98.99 \\ 168 99.26 \\ 171 99.76 \\ 174 100.2 \\ 177 100.7 \\ 180 100.9 \\ 183 101.4 \\ 186 101.9 \\ 189 102.1 \\ 192 102.5 \\ 195 103.0 \\ 198 103.5 \\ 201 103.7 \\ 204 104.1 \\ 207 104.3 \\ 210 104.7 \\ 213 105.1 \\ 216 105.4 \\ 219 105.8 \\ 222 106.2 \\ 225 106.6 \\ 228 106.7 \\ 231 107.1 \\ 234 107.3 \\ 237 107.5 \\ 240 107.9 \\ 243 108.2 \\ 246 108.6 \\ 249 109.0 \\ 252 108.9 \\ 255 109.3 \\ 258 109.7 \\ 261 110.0 \\ 264 110.4 \\ 267 110.5 \\ 270 110.9 \\ 273 111.2 \\ 276 111.3 \\ 279 111.7 \\ 282 112.1 \\ 285 112.2 \\ 288 112.3 \\ 291 112.7 \\ 294 113.0 \\ 297 113.3 \\ 300 113.4 \\ 303 113.8 \\ 306 113.9 \\ 309 114.2 \\ 312 114.3 \\ 315 114.6 \\ 318 114.9 \\ 321 115.1 \\ 324 115.4 \\ 327 115.7 \\ 330 115.8 \\ 333 116.1 \\ 336 116.2 \\ 339 116.5 \\ 342 116.6 \\ 345 116.9 \\ 348 117.0 \\ 351 117.3 \\ 354 117.6 \\ 357 117.7 \\ 360 118.0 \\ 363 118.2 \\ 366 118.3 \\ 369 118.6 \\ 372 118.7 \\ 375 119.0 \\ 378 119.1 \\ 381 119.4 \\ 384 119.5 \\ 387 119.7 \\ 390 120.0 \\ 393 120.1 \\ 396 120.3 \\ 399 120.5 \\ 402 120.7 \\ 405 121.0 \\ 408 121.0 \\ 411 121.3 \\ 414 121.4 \\ 417 121.6 \\ 420 121.9 \\ 423 122.0 \\ 426 122.1 \\ 429 122.3 \\ 432 122.5 \\ 435 122.6 \\ 438 122.9 \\ 441 123.1 \\ 444 123.2 \\ 447 123.4 \\ 450 123.5 \\ 453 123.7 \\ 456 123.8 \\ 459 124.0 \\ 462 124.1 \\ 465 124.4 \\ 468 124.6 \\ 471 124.7 \\ 474 124.9 \\ 477 125.0 \\ 480 125.2 \\ 483 125.4 \\ 486 125.5 \\ 489 125.6 \\ 492 125.8 \\ 495 126.0 \\ 498 126.1 \\ 
         }{\dataTheoOptSeptRQC}

        \pgfplotstableread[row sep=\\]{
        X Y\\ 21 75.90 \\ 24 79.79 \\ 27 85.25 \\ 30 88.49 \\ 33 87.11 \\ 36 90.73 \\ 39 86.79 \\ 42 89.62 \\ 45 91.37 \\ 48 87.27 \\ 51 88.70 \\ 54 84.17 \\ 57 84.75 \\ 60 86.65 \\ 63 87.74 \\ 66 88.95 \\ 69 89.98 \\ 72 91.02 \\ 75 92.12 \\ 78 93.05 \\ 81 94.04 \\ 84 94.92 \\ 87 95.82 \\ 90 96.78 \\ 93 97.59 \\ 96 98.46 \\ 99 99.24 \\ 102 100.0 \\ 105 100.8 \\ 108 101.2 \\ 111 101.9 \\ 114 103.0 \\ 117 103.3 \\ 120 104.4 \\ 123 104.7 \\ 126 105.4 \\ 129 106.1 \\ 132 106.7 \\ 135 106.7 \\ 138 107.7 \\ 141 108.3 \\ 144 108.9 \\ 147 109.5 \\ 150 110.0 \\ 153 110.6 \\ 156 110.9 \\ 159 111.8 \\ 162 112.3 \\ 165 112.8 \\ 168 113.1 \\ 171 113.6 \\ 174 114.1 \\ 177 114.6 \\ 180 114.8 \\ 183 115.3 \\ 186 116.1 \\ 189 116.3 \\ 192 116.8 \\ 195 117.3 \\ 198 117.7 \\ 201 117.9 \\ 204 118.4 \\ 207 118.8 \\ 210 119.3 \\ 213 119.7 \\ 216 119.9 \\ 219 120.3 \\ 222 120.8 \\ 225 121.2 \\ 228 121.6 \\ 231 121.7 \\ 234 122.1 \\ 237 122.4 \\ 240 122.7 \\ 243 123.1 \\ 246 123.5 \\ 249 123.9 \\ 252 124.1 \\ 255 124.4 \\ 258 124.8 \\ 261 125.2 \\ 264 125.5 \\ 267 125.9 \\ 270 126.3 \\ 273 126.4 \\ 276 126.7 \\ 279 127.1 \\ 282 127.5 \\ 285 127.8 \\ 288 128.0 \\ 291 128.3 \\ 294 128.6 \\ 297 129.0 \\ 300 129.1 \\ 303 129.4 \\ 306 129.5 \\ 309 129.9 \\ 312 130.2 \\ 315 130.5 \\ 318 130.8 \\ 321 131.0 \\ 324 131.3 \\ 327 131.6 \\ 330 131.7 \\ 333 132.0 \\ 336 132.1 \\ 339 132.4 \\ 342 132.7 \\ 345 133.0 \\ 348 133.1 \\ 351 133.4 \\ 354 133.7 \\ 357 134.0 \\ 360 134.3 \\ 363 134.5 \\ 366 134.6 \\ 369 134.9 \\ 372 135.2 \\ 375 135.5 \\ 378 135.4 \\ 381 135.7 \\ 384 136.0 \\ 387 136.2 \\ 390 136.5 \\ 393 136.6 \\ 396 136.8 \\ 399 137.0 \\ 402 137.2 \\ 405 137.5 \\ 408 137.7 \\ 411 138.0 \\ 414 138.2 \\ 417 138.3 \\ 420 138.6 \\ 423 138.8 \\ 426 138.9 \\ 429 139.2 \\ 432 139.4 \\ 435 139.5 \\ 438 139.7 \\ 441 140.0 \\ 444 140.1 \\ 447 140.3 \\ 450 140.6 \\ 453 140.8 \\ 456 140.9 \\ 459 141.1 \\ 462 141.2 \\ 465 141.4 \\ 468 141.6 \\ 471 141.9 \\ 474 142.1 \\ 477 142.2 \\ 480 142.4 \\ 483 142.6 \\ 486 142.7 \\ 489 142.9 \\ 492 143.0 \\ 495 143.2 \\ 498 143.3 \\ 
         }{\dataTheoOptHuitRQC}

        \pgfplotstableread[row sep=\\]{
        X Y\\ 21 82.93 \\ 24 91.61 \\ 27 98.28 \\ 30 102.1 \\ 33 107.2 \\ 36 111.8 \\ 39 109.2 \\ 42 112.6 \\ 45 114.7 \\ 48 117.6 \\ 51 120.4 \\ 54 115.5 \\ 57 109.4 \\ 60 111.5 \\ 63 113.5 \\ 66 107.0 \\ 69 108.7 \\ 72 109.8 \\ 75 103.0 \\ 78 104.5 \\ 81 105.5 \\ 84 106.4 \\ 87 107.9 \\ 90 108.3 \\ 93 109.6 \\ 96 110.5 \\ 99 111.3 \\ 102 112.2 \\ 105 112.9 \\ 108 113.3 \\ 111 114.5 \\ 114 115.2 \\ 117 116.0 \\ 120 117.1 \\ 123 117.4 \\ 126 118.5 \\ 129 119.1 \\ 132 119.8 \\ 135 119.8 \\ 138 120.4 \\ 141 121.4 \\ 144 122.0 \\ 147 123.0 \\ 150 123.5 \\ 153 124.1 \\ 156 124.4 \\ 159 125.3 \\ 162 125.8 \\ 165 126.4 \\ 168 126.6 \\ 171 127.5 \\ 174 128.0 \\ 177 128.5 \\ 180 128.7 \\ 183 129.5 \\ 186 130.0 \\ 189 130.5 \\ 192 131.0 \\ 195 131.5 \\ 198 131.9 \\ 201 132.4 \\ 204 132.9 \\ 207 133.1 \\ 210 133.8 \\ 213 134.3 \\ 216 134.5 \\ 219 134.9 \\ 222 135.3 \\ 225 135.7 \\ 228 136.1 \\ 231 136.6 \\ 234 137.0 \\ 237 137.2 \\ 240 137.6 \\ 243 138.0 \\ 246 138.6 \\ 249 139.0 \\ 252 139.2 \\ 255 139.6 \\ 258 139.9 \\ 261 140.3 \\ 264 140.9 \\ 267 141.0 \\ 270 141.4 \\ 273 141.8 \\ 276 142.1 \\ 279 142.5 \\ 282 142.8 \\ 285 143.2 \\ 288 143.4 \\ 291 143.7 \\ 294 144.0 \\ 297 144.6 \\ 300 144.7 \\ 303 145.0 \\ 306 145.4 \\ 309 145.7 \\ 312 145.8 \\ 315 146.4 \\ 318 146.7 \\ 321 146.8 \\ 324 147.1 \\ 327 147.4 \\ 330 147.7 \\ 333 148.0 \\ 336 148.2 \\ 339 148.5 \\ 342 148.8 \\ 345 149.1 \\ 348 149.2 \\ 351 149.5 \\ 354 149.8 \\ 357 150.1 \\ 360 150.5 \\ 363 150.8 \\ 366 150.9 \\ 369 151.2 \\ 372 151.5 \\ 375 151.8 \\ 378 151.9 \\ 381 152.2 \\ 384 152.4 \\ 387 152.7 \\ 390 153.0 \\ 393 153.2 \\ 396 153.5 \\ 399 153.6 \\ 402 153.9 \\ 405 154.1 \\ 408 154.4 \\ 411 154.6 \\ 414 154.9 \\ 417 155.2 \\ 420 155.3 \\ 423 155.5 \\ 426 155.8 \\ 429 156.0 \\ 432 156.3 \\ 435 156.4 \\ 438 156.6 \\ 441 157.0 \\ 444 157.1 \\ 447 157.3 \\ 450 157.6 \\ 453 157.8 \\ 456 157.9 \\ 459 158.1 \\ 462 158.4 \\ 465 158.6 \\ 468 158.8 \\ 471 158.9 \\ 474 159.3 \\ 477 159.4 \\ 480 159.6 \\ 483 159.8 \\ 486 159.9 \\ 489 160.1 \\ 492 160.4 \\ 495 160.6 \\ 498 160.7 \\ 
         }{\dataTheoOptNeufRQC}

        \pgfplotstableread[row sep=\\]{
          X Y \\ 20 30.51 \\ 22 31.29 \\ 24 32.09 \\ 26 32.82 \\ 28 33.54 \\ 30 34.18 \\ 32 34.75 \\ 34 35.31 \\ 36 35.87 \\ 38 36.38 \\ 40 36.84 \\ 42 37.29 \\ 44 37.73 \\ 46 38.21 \\ 48 38.59\\50 39.01\\52 39.40\\54 39.72\\ 56 40.08\\58 40.42\\60 40.76\\62 41.06\\64 41.42\\66 41.74\\68 42.03\\70 42.32\\72 42.67\\74 42.96\\76 43.25\\78 43.53\\80 43.80\\82 44.14\\84 44.36\\86 44.59\\88 44.93\\90 45.15\\92 45.41\\94 45.64\\96 45.86\\%98 46.39\\
        }\dataExpOptTrois

        \pgfplotstableread[row sep=\\]{
          X Y \\ 20 33.68 \\ 22 34.60 \\ 24 35.53 \\ 26 36.47 \\ 28 37.34 \\ 30 38.24 \\ 32 39.11 \\ 34 39.92 \\ 36 40.72 \\ 38 41.50 \\ 40 42.28 \\42 43.05\\44 43.78\\46 44.70\\48 45.32\\ 50 46.13\\52 46.81\\54 47.45\\%56 47.98\\
        }\dataExpOptQuatre

        \pgfplotstableread[row sep=\\]{
          X Y\\ 20 35.41\\ 22 37.14\\ 24 38.80\\ 26 40.34\\ 28 41.70\\ 30 43.47\\ 32 44.89\\ 34 46.18\\ 36 47.30\\ %38 47.6\\
        }\dataExpOptCinq

        \pgfplotstableread[row sep=\\]{
        X Y \\ 20 45.59\\ 22 48.85\\ 24 51.62\\ 26 54.66\\ 28 57.08\\ 30 59.07\\ 32 60.93\\ 34 62.67\\ 36 64.65\\ 38 66.29\\ 40 67.76\\ 42 69.15\\ 44 70.87\\ 46 72.00\\ 48 73.21\\ 50 74.52\\ 52 75.70\\ 54 76.77\\ 56 78.13\\ 58 78.93\\ 60 79.89\\ 62 81.12\\ 64 81.90\\ 66 82.77\\ 68 83.61\\ 70 84.70\\ 72 85.39\\ 74 86.21\\ 76 86.97\\ 78 87.70\\ 80 88.56\\ 82 89.26\\ 84 89.94\\ 86 90.69\\ 88 91.38\\ 90 92.01\\ 92 92.64\\ 94 93.15\\ 96 93.75\\ 98 94.34\\ 100 94.99\\ 102 95.55\\ 104 96.10\\ 106 96.75\\ 108 97.28\\ 110 97.80\\ 112 98.31\\ 114 98.82\\ 116 99.27\\ 118 99.76\\ 120 100.29\\ 122 100.77\\ 124 101.26\\ 126 101.72\\ 128 102.22\\ 130 102.54\\ 132 103.13\\ 134 103.46\\ 136 103.88\\ 138 104.30\\ 140 104.72\\ 142 105.12\\ 144 105.52\\ 146 105.92\\ 148 106.31\\ 150 106.69\\ 152 107.11\\ 154 107.49\\ 156 107.90\\ 158 108.27\\ 160 108.63\\ 162 108.99\\ 164 109.28\\ 166 109.65\\ 168 110.00\\ 170 110.34\\ 172 110.67\\ 174 111.01\\ 176 111.34\\ 178 111.65\\ 180 111.99\\ 182 112.31\\ 184 112.63\\ 186 112.88\\ 188 113.19\\ 190 113.50\\ 192 113.80\\ 194 114.16\\ 196 114.45\\ 198 114.78\\ 200 115.08\\ 
        }{\dataTheoSept}

        \pgfplotstableread[row sep=\\]{
          X Y\\20 42.69\\ 22 45.51\\ 24 47.88\\ 26 49.91\\ 28 52.05\\ 30 53.77\\ 32 55.87\\ 34 57.36\\ 36 58.65\\ 38 60.09\\ 40 61.36\\ 42 62.57\\ 44 63.72\\ 46 65.07\\ 48 66.12\\ 50 67.27\\ 52 67.99\\ 54 69.24\\ 56 70.13\\ 58 70.83\\ 60 71.67\\ 62 72.49\\ 64 73.45\\ 66 74.21\\ 68 74.95\\ 70 75.66\\ 72 76.27\\ 74 77.01\\ 76 77.67\\ 78 78.32\\ 80 78.86\\ 82 79.48\\ 84 80.08\\ 86 80.76\\ 88 81.37\\ 90 81.93\\ 92 82.48\\ 94 82.94\\ 96 83.47\\ 98 83.99\\ 100 84.57\\ 102 85.07\\ 104 85.55\\ 106 85.96\\ 108 86.43\\ 110 86.89\\ 112 87.35\\ 114 87.79\\ 116 88.35\\ 118 88.78\\ 120 89.26\\ 122 89.68\\ 124 90.11\\ 126 90.52\\ 128 90.97\\ 130 91.38\\ 132 91.76\\ 134 92.19\\ 136 92.56\\ 138 92.93\\ 140 93.29\\ 142 93.64\\ 144 93.99\\ 146 94.33\\ 148 94.67\\ 150 95.01\\ 152 95.51\\ 154 95.83\\ 156 96.20\\ 158 96.52\\ 160 96.83\\ 162 97.14\\ 164 97.51\\ 166 97.83\\ 168 98.13\\ 170 98.42\\ 172 98.71\\ 174 99.00\\ 176 99.28\\ 178 99.66\\ 180 99.96\\ 182 100.23\\ 184 100.51\\ 186 100.83\\ 188 101.10\\ 190 101.36\\ 192 101.62\\ 194 101.93\\ 196 102.18\\ 198 102.47\\ 200 102.72\\ 
         }{\dataTheoSix}
          
        \pgfplotstableread[row sep=\\]{
        X Y\\  20 39.09\\ 22 41.49\\ 24 43.48\\ 26 45.15\\ 28 46.47\\ 30 48.41\\ 32 49.74\\ 34 51.00\\ 36 52.12\\ 38 53.35\\ 40 54.43\\ 42 55.45\\ 44 56.42\\ 46 57.28\\ 48 58.18\\ 50 58.88\\ 52 59.79\\ 54 60.59\\ 56 61.37\\ 58 61.97\\ 60 62.70\\ 62 63.40\\ 64 64.25\\ 66 64.91\\ 68 65.55\\ 70 66.17\\ 72 66.92\\ 74 67.55\\ 76 68.11\\ 78 68.66\\ 80 69.34\\ 82 69.86\\ 84 70.37\\ 86 70.94\\ 88 71.47\\ 90 71.94\\ 92 72.40\\ 94 72.97\\ 96 73.41\\ 98 73.84\\ 100 74.34\\ 102 74.75\\ 104 75.16\\ 106 75.66\\ 108 76.05\\ 110 76.43\\ 112 76.80\\ 114 77.17\\ 116 77.66\\ 118 78.02\\ 120 78.42\\ 122 78.77\\ 124 79.13\\ 126 79.47\\ 128 79.85\\ 130 80.20\\ 132 80.51\\ 134 80.88\\ 136 81.19\\ 138 81.49\\ 140 81.79\\ 142 82.08\\ 144 82.37\\ 146 82.65\\ 148 82.94\\ 150 83.21\\ 152 83.66\\ 154 83.93\\ 156 84.24\\ 158 84.50\\ 160 84.76\\ 162 85.02\\ 164 85.33\\ 166 85.61\\ 168 85.85\\ 170 86.10\\ 172 86.34\\ 174 86.58\\ 176 86.81\\ 178 87.14\\ 180 87.39\\ 182 87.62\\ 184 87.85\\ 186 88.13\\ 188 88.35\\ 190 88.56\\ 192 88.78\\ 194 89.05\\ 196 89.26\\ 198 89.50\\ 200 89.71\\ 
         }{\dataTheoCinq}
         
        \pgfplotstableread[row sep=\\]{
        X Y\\ 10 21.82\\ 12 25.90\\ 14 28.38\\ 16 30.91\\ 18 32.91\\ 20 34.69\\ 22 36.14\\ 24 37.81\\ 26 39.18\\ 28 40.27\\ 30 41.47\\ 32 42.59\\ 34 43.64\\ 36 44.60\\ 38 45.64\\ 40 46.54\\ 42 47.39\\ 44 48.19\\ 46 49.22\\ 48 49.95\\ 50 50.80\\ 52 51.54\\ 54 52.18\\ 56 52.80\\ 58 53.53\\ 60 54.10\\ 62 54.66\\ 64 55.37\\ 66 55.89\\ 68 56.39\\ 70 56.88\\ 72 57.51\\ 74 58.02\\ 76 58.47\\ 78 58.90\\ 80 59.46\\ 82 59.87\\ 84 60.28\\ 86 60.75\\ 88 61.18\\ 90 61.55\\ 92 61.92\\ 94 62.39\\ 96 62.74\\ 98 63.09\\ 100 63.49\\102 63.82\\ 104 64.14\\ 106 64.56\\ 108 64.87\\ 110 65.18\\ 112 65.47\\ 114 65.77\\ 116 66.18\\ 118 66.46\\ 120 66.80\\ 122 67.07\\ 124 67.37\\ 126 67.63\\ 128 67.95\\ 130 68.23\\ 132 68.48\\ 134 68.78\\ 136 69.03\\ 138 69.27\\ 140 69.51\\ 142 69.74\\ 144 69.97\\ 146 70.20\\ 148 70.42\\ 150 70.64\\ 152 71.03\\ 154 71.25\\ 156 71.50\\ 158 71.71\\ 160 71.92\\ 162 72.12\\ 164 72.39\\ 166 72.61\\ 168 72.81\\ 170 73.00\\ 172 73.19\\ 174 73.38\\ 176 73.57\\ 178 73.86\\ 180 74.06\\ 182 74.24\\ 184 74.42\\ 186 74.65\\ 188 74.83\\ 190 75.00\\ 192 75.17\\ 194 75.40\\ 196 75.57\\ 198 75.77\\ 200 75.93\\ 
        }{\dataTheoQuatre}

        \pgfplotstableread[row sep=\\]{
        X Y\\
        10 19.56 \\ 12 21.94 \\ 14 23.90 \\ 16 25.94 \\ 18 27.50 \\ 20 28.89 \\ 22 30.05 \\ 24 31.38 \\ 26 32.45 \\ 28 33.70 \\ 30 34.61 \\ 32 35.45 \\ 34 36.24 \\ 36 37.31 \\ 38 38.10 \\ 40 38.76 \\ 42 39.39 \\ 44 39.98 \\ 46 40.80 \\ 48 41.34 \\ 50 42.01 \\ 52 42.57 \\ 54 43.05 \\ 56 43.51 \\ 58 44.08 \\ 60 44.50 \\ 62 44.91 \\ 64 45.49 \\ 66 45.87 \\ 68 46.24 \\ 70 46.61 \\ 72 47.11 \\ 74 47.50 \\ 76 47.84 \\ 78 48.16 \\ 80 48.61 \\ 82 48.92 \\ 84 49.22 \\ 86 49.59 \\ 88 49.92 \\ 90 50.20 \\ 92 50.47 \\ 94 50.85 \\ 96 51.11 \\ 98 51.37 \\ 100 51.69 \\ 102 51.93\\ 104 52.17\\ 106 52.51\\ 108 52.74\\ 110 52.97\\ 112 53.19\\ 114 53.41\\ 116 53.75\\ 118 53.96\\ 120 54.23\\ 122 54.43\\ 124 54.66\\ 126 54.86\\ 128 55.10\\ 130 55.32\\ 132 55.51\\ 134 55.75\\ 136 55.93\\ 138 56.11\\ 140 56.29\\ 142 56.46\\ 144 56.63\\ 146 56.80\\ 148 56.97\\ 150 57.13\\ 152 57.47\\ 154 57.63\\ 156 57.83\\ 158 57.99\\ 160 58.14\\ 162 58.29\\ 164 58.51\\ 166 58.68\\ 168 58.83\\ 170 58.97\\ 172 59.11\\ 174 59.26\\ 176 59.40\\ 178 59.63\\ 180 59.79\\ 182 59.93\\ 184 60.06\\ 186 60.25\\ 188 60.38\\ 190 60.51\\ 192 60.64\\ 194 60.82\\ 196 60.95\\ 198 61.10\\ 200 61.23\\ 
        }{\dataTheoTrois}

        \pgfplotstableread[row sep=\\]{
        X Y \\ 186 112.88\\ 188 113.19\\ 190 113.50\\ 192 113.80\\ 194 114.16\\ 196 114.45\\ 198 114.78\\ 200 115.08\\ 
        }{\dataTheoSeptKN}

        \pgfplotstableread[row sep=\\]{
          X Y\\106 85.96\\ 108 86.43\\ 110 86.89\\ 112 87.35\\ 114 87.79\\ 116 88.35\\ 118 88.78\\ 120 89.26\\ 122 89.68\\ 124 90.11\\ 126 90.52\\ 128 90.97\\ 130 91.38\\ 132 91.76\\ 134 92.19\\ 136 92.56\\ 138 92.93\\ 140 93.29\\ 142 93.64\\ 144 93.99\\ 146 94.33\\ 148 94.67\\ 150 95.01\\ 152 95.51\\ 154 95.83\\ 156 96.20\\ 158 96.52\\ 160 96.83\\ 162 97.14\\ 164 97.51\\ 166 97.83\\ 168 98.13\\ 170 98.42\\ 172 98.71\\ 174 99.00\\ 176 99.28\\ 178 99.66\\ 180 99.96\\ 182 100.23\\ 184 100.51\\ 186 100.83\\ 188 101.10\\ 190 101.36\\ 192 101.62\\ 194 101.93\\ 196 102.18\\ 198 102.47\\ 200 102.72\\ 
         }{\dataTheoSixKN}
          
        \pgfplotstableread[row sep=\\]{
        X Y\\    58 61.97\\ 60 62.70\\ 62 63.40\\ 64 64.25\\ 66 64.91\\ 68 65.55\\ 70 66.17\\ 72 66.92\\ 74 67.55\\ 76 68.11\\ 78 68.66\\ 80 69.34\\ 82 69.86\\ 84 70.37\\ 86 70.94\\ 88 71.47\\ 90 71.94\\ 92 72.40\\ 94 72.97\\ 96 73.41\\ 98 73.84\\ 100 74.34\\ 102 74.75\\ 104 75.16\\ 106 75.66\\ 108 76.05\\ 110 76.43\\ 112 76.80\\ 114 77.17\\ 116 77.66\\ 118 78.02\\ 120 78.42\\ 122 78.77\\ 124 79.13\\ 126 79.47\\ 128 79.85\\ 130 80.20\\ 132 80.51\\ 134 80.88\\ 136 81.19\\ 138 81.49\\ 140 81.79\\ 142 82.08\\ 144 82.37\\ 146 82.65\\ 148 82.94\\ 150 83.21\\ 152 83.66\\ 154 83.93\\ 156 84.24\\ 158 84.50\\ 160 84.76\\ 162 85.02\\ 164 85.33\\ 166 85.61\\ 168 85.85\\ 170 86.10\\ 172 86.34\\ 174 86.58\\ 176 86.81\\ 178 87.14\\ 180 87.39\\ 182 87.62\\ 184 87.85\\ 186 88.13\\ 188 88.35\\ 190 88.56\\ 192 88.78\\ 194 89.05\\ 196 89.26\\ 198 89.50\\ 200 89.71\\ 
         }{\dataTheoCinqKN}
         
        \pgfplotstableread[row sep=\\]{
        X Y\\ 36 44.60\\ 38 45.64\\ 40 46.54\\ 42 47.39\\ 44 48.19\\ 46 49.22\\ 48 49.95\\ 50 50.80\\ 52 51.54\\ 54 52.18\\ 56 52.80\\ 58 53.53\\ 60 54.10\\ 62 54.66\\ 64 55.37\\ 66 55.89\\ 68 56.39\\ 70 56.88\\ 72 57.51\\ 74 58.02\\ 76 58.47\\ 78 58.90\\ 80 59.46\\ 82 59.87\\ 84 60.28\\ 86 60.75\\ 88 61.18\\ 90 61.55\\ 92 61.92\\ 94 62.39\\ 96 62.74\\ 98 63.09\\ 100 63.49\\102 63.82\\ 104 64.14\\ 106 64.56\\ 108 64.87\\ 110 65.18\\ 112 65.47\\ 114 65.77\\ 116 66.18\\ 118 66.46\\ 120 66.80\\ 122 67.07\\ 124 67.37\\ 126 67.63\\ 128 67.95\\ 130 68.23\\ 132 68.48\\ 134 68.78\\ 136 69.03\\ 138 69.27\\ 140 69.51\\ 142 69.74\\ 144 69.97\\ 146 70.20\\ 148 70.42\\ 150 70.64\\ 152 71.03\\ 154 71.25\\ 156 71.50\\ 158 71.71\\ 160 71.92\\ 162 72.12\\ 164 72.39\\ 166 72.61\\ 168 72.81\\ 170 73.00\\ 172 73.19\\ 174 73.38\\ 176 73.57\\ 178 73.86\\ 180 74.06\\ 182 74.24\\ 184 74.42\\ 186 74.65\\ 188 74.83\\ 190 75.00\\ 192 75.17\\ 194 75.40\\ 196 75.57\\ 198 75.77\\ 200 75.93\\ 
        }{\dataTheoQuatreKN}

        \pgfplotstableread[row sep=\\]{
        X Y\\22 30.05 \\ 24 31.38 \\ 26 32.45 \\ 28 33.70 \\ 30 34.61 \\ 32 35.45 \\ 34 36.24 \\ 36 37.31 \\ 38 38.10 \\ 40 38.76 \\ 42 39.39 \\ 44 39.98 \\ 46 40.80 \\ 48 41.34 \\ 50 42.01 \\ 52 42.57 \\ 54 43.05 \\ 56 43.51 \\ 58 44.08 \\ 60 44.50 \\ 62 44.91 \\ 64 45.49 \\ 66 45.87 \\ 68 46.24 \\ 70 46.61 \\ 72 47.11 \\ 74 47.50 \\ 76 47.84 \\ 78 48.16 \\ 80 48.61 \\ 82 48.92 \\ 84 49.22 \\ 86 49.59 \\ 88 49.92 \\ 90 50.20 \\ 92 50.47 \\ 94 50.85 \\ 96 51.11 \\ 98 51.37 \\ 100 51.69 \\ 102 51.93\\ 104 52.17\\ 106 52.51\\ 108 52.74\\ 110 52.97\\ 112 53.19\\ 114 53.41\\ 116 53.75\\ 118 53.96\\ 120 54.23\\ 122 54.43\\ 124 54.66\\ 126 54.86\\ 128 55.10\\ 130 55.32\\ 132 55.51\\ 134 55.75\\ 136 55.93\\ 138 56.11\\ 140 56.29\\ 142 56.46\\ 144 56.63\\ 146 56.80\\ 148 56.97\\ 150 57.13\\ 152 57.47\\ 154 57.63\\ 156 57.83\\ 158 57.99\\ 160 58.14\\ 162 58.29\\ 164 58.51\\ 166 58.68\\ 168 58.83\\ 170 58.97\\ 172 59.11\\ 174 59.26\\ 176 59.40\\ 178 59.63\\ 180 59.79\\ 182 59.93\\ 184 60.06\\ 186 60.25\\ 188 60.38\\ 190 60.51\\ 192 60.64\\ 194 60.82\\ 196 60.95\\ 198 61.10\\ 200 61.23\\ 
        }{\dataTheoTroisKN}

        In this section, we show that, when the number of equations is
        sufficiently large, we can solve the system given in Modeling \ref{mod:system} with only
        linear algebra computations, by linearization on the $c_T$'s.
        %polynomials $\det(\Cm_{\any,{\mat T}})$.

        \subsection{The overdetermined case}
%        The system $\MaxMinors$ can be viewed as a linear system with
%        $m\binom{n-k-1}{r}$ linear equations over $\Fq$, in the
%        $\binom{n}{r}-1$ variables $c_T$ representing the non constant
%        polynomials $\det(\Cm_{\any,{\mat T}})$, for all
%        ${\mat T}\subset\{1..n\}$, $\size{\mat T}=r$, ${\mat T}\ne
%        \{1..r\}$. According to~\cref{rem:ct}, if we are able to linearise this sytem with respect to the variables $c_T$, then in particular we get the values of all the entries $c_{i,j}$ of the matrix $\Cm$.
%
%        In order to linearise this system, we can expand each equation over $\Fqm$ as
%        $m$ equations over $\Fq$, and construct a matrix 
%        $\MaxMin$ 
        The linear system given in Modeling \ref{mod:system} is described by the following matrix $\MaxMin$ with rows
        indexed by $(J,i) : J\subset\{1..n-k-1\}, \size{J}=r, 0\le i \le m-1$ and
        columns indexed by $T \subset\{1..n\}$ of size $r$, with the
        entry in row $(J,i)$ and column $T$ being the coefficient in
        $\alpha^i$ of the element
        $\pm\minor{\Rm}{T_1,J\backslash T_2} \in \Fqm$. More
        precisely, we have
        \begin{eqnarray}
          \MaxMin[(J,i), T]
          &=&
          \begin{cases}
            0 & \text{ if } T_2 \not \subset J\\
            [\alpha^i](-1)^{\Sign{T_2}{J}}(\minor{\Rm}{T_1, J\backslash T_2}) & \text{ if } T_2 \subset J,
          \end{cases}\\
          \text{ with } && T_1=  T \cap \{1..k+1\},\notag{}\\
          \text{ and } && T_2= ( T \cap \{k+2..n\})-(k+1).\notag{}
        \end{eqnarray}
        The matrix $\MaxMin$ can have rank $\binom n r -1$ at most; indeed if it 
        had a maximal rank of $\binom{n}{r}$, this would imply that all $c_T$'s are equal to 
        $0$, which is in contradiction with the assumption $c_{\{1..r\}}=1$.
        %$\langle \MaxMinors(\Cm\trsp{\Hm})\rangle = \langle 1\rangle$.
        \begin{proposition}
          If $\MaxMin$ has rank $\binom n r - 1$, 
%(which implies that    $m\binom{n-k-1}r\ge \binom n r - 1$), 
then the right kernel of
          $\MaxMin$ contains only one element $
          \begin{pmatrix}
            \cv & 1
          \end{pmatrix}\in\Fq^{\binom n r}$ with value 1 on its component
          corresponding to $c_{\{1..r\}}$. The components of
          this vector contain the values of the $c_{T}$'s, 
          %the $\det(\Cm_{\any,\mat T})$,
        $T\ne\{1..r\}$. This gives the values of all the
          variables
          $C_{i,j} = (-1)^{r+i} c_{\{1..r\}\backslash\{i\}\cup\{j\}}$.
        \end{proposition}
        \begin{proof}
          If $\MaxMin$ has rank $\binom{n}{r}-1$, then as there is a solution
          to the system, a row echelon form of the matrix has the shape
         \begin{equation*}
           \begin{pmatrix}
           \Im_{\binom n r - 1} & -\trsp{\cv}\\
           \zerom & \zerom
         \end{pmatrix}
         \end{equation*}
         with $\cv$ a vector in $\Fq$ of size $\binom n r - 1$: we cannot get
         a jump in the stair of the echelon form as it would imply
         that~$\mathcal F_M$
         %\cref{eq:maxminors-system} 
         has no solution. Then $
         \begin{pmatrix}
           \cv & 1
         \end{pmatrix}$ is in the right kernel of $\MaxMin$. \qed
        \end{proof}
        It is then easy to recover the variables $\Sm$
        from~\eqref{eqn:oj-modelling-specialized-in-C} by linear
        algebra.  The following algorithm recovers the error if there
        is one solution to the
        system~\eqref{eqn:oj-modelling-specialized-in-C}. It is shown
        in~\cite[Algorithm 1]{BBBGNRT19} how to deal with the other
        cases, and this can be easily adapted to the 
        specialization considered in this paper.

         \begin{algorithm}[htbp]
          \caption{$(m,n,k,r)$-decoding in the overdetermined case.}
         \KwData{Code $\mathcal C$, vector $\yv$ at distance $r$ from $\mathcal C$, such that $m\binom{n-k-1}r \ge \binom n r - 1$ and $\MaxMin$ has rank $\binom n r -1$}
         \KwResult{The error $\ev$ of weight $r$  such that $\yv-\ev\in\mathcal C$}
         Construct  \(\MaxMin\), the \(m\binom{n-k-1}r \times \binom n r\) matrix over $\ff q$ associated to the system 
%$\MaxMinors$ 
$\mathcal F_M$ \; 
         %\cref{eq:maxminors-system} \;    
        Let $
        \begin{pmatrix}
        \cv & 1
        \end{pmatrix}
        $ be the only such vector in the right kernel of $\MaxMin$ \;
          {Compute the values $\Cm^* = (C^*_{i,j})_{i,j}$ from $\cv$\;
            Compute the values $(S_1^*,\dots,S_r^*)\in\ff {q^m}^r$ by solving the linear system\\
          \[(S_1,\dots,S_r)\cm^*
        \trsp{\Hm} = 0\] and taking the unique value with $S_1^*=1$\;
        }
          \Return $(1,S_2^*,\dots,S_r^*)\cm^*$ \;
         \label{alg:overdetermined-case}
        \end{algorithm}
        \begin{proposition}
        When $m\binom{n-k-1}r \ge \binom n r - 1$ and $\MaxMin$ has rank $\binom n r -1$, 
        then~\cref{alg:overdetermined-case} recovers the error in complexity 
         \begin{equation}\label{eq:overdetermined-complexity}
        \mathcal{O}\left(   m\binom{n-k-1}r {\binom n r}^{\omega-1}\right)
        \end{equation}
        operations in the field $\ff q$, where $\omega$ is the constant of
        linear algebra.
        \end{proposition}
        \begin{proof}
          To recover the error, the most consuming part is the computation of
          the left kernel of the matrix $\MaxMin$ in
          $\ff q^{m\binom{n-k-1}r\times \binom n r}$, in the case where
          $m\binom{n-k-1}r \ge \binom n r - 1$. % This can be done by computing
          % an echelon form of $\MaxMin$, in this case the
          This complexity is bounded
          by~\cref{eq:overdetermined-complexity}. \qed
        \end{proof}

         We ran a lot of experiments with random codes $\mathcal C$ such that 
         $m\binom{n-k-1}r\ge\binom n r -1$, and the matrix $\MaxMin$ was always 
         of rank $\binom n r -1$. 
          That is why we propose the
         following heuristic about the rank of $\MaxMin$.
        \begin{heuristic}[Overdetermined case]\label{heuristic:over-determined}
          When $m\binom{n-k-1}r\ge\binom{n}r-1$, with overwhelming
          probability, the rank of the matrix $\MaxMin$ is $\binom n r - 1$.
        \end{heuristic}

        Figure~\ref{fig:r=345exp} gives the experimental results  for $q=2$, $r=3, 4, 5$ and different values of
        $n$. We choose to keep $m$ prime and close to $n/1.18$ to have
        a data set containing the parameters of the ROLLO-I
        cryptosystem. We choose for $k$ the minimum between
        $\frac n 2$ and the largest value leading to an overdetermined
        case. We have $k=\frac n 2$ as soon as $n\ge 22$ for $r=3$,
        $n\ge 36$ for $r=4$, $n\ge 58$ for $r=5$. The figure shows
        that the estimated complexity is a good upper bound for the
        computation's complexity. It also shows that this upper bound
        is not tight. Note that the experimental values are the
        complexity of the whole attack, including building the
        matrix that requires to compute the minors of $\Rm$. Hence for
        small values of $n$, it may happen that this part of the
        attack takes more time than solving the linear
        system. This explains why, for $r=3$ and $n<28$, the
        experimental curve is above the theoretical one.

        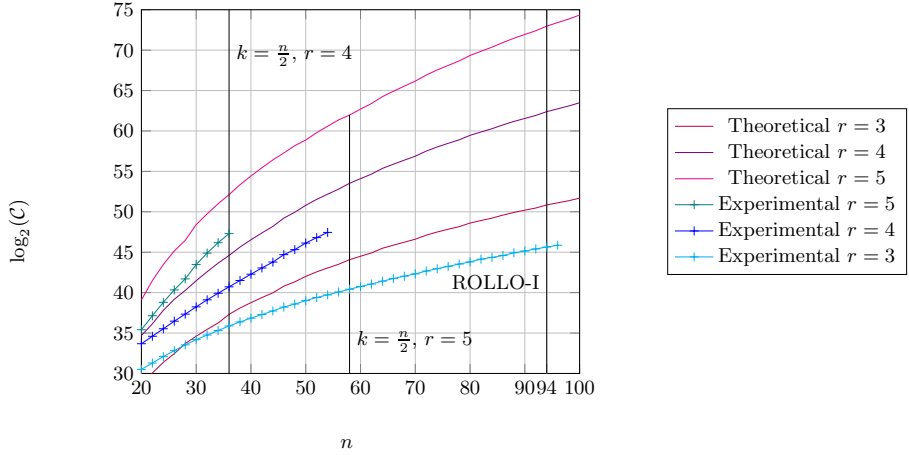
\begin{figure}[h]
          \centering
          \begin{tikzpicture}[scale=0.85]
           \pgfplotsset{every axis legend/.append style={at={(1,0.2)},anchor=east}}
           \def\nmin{20}\def\nmax{100}
           \def\Cmin{30}\def\Cmax{75}
           \begin{axis}[ grid=major, extra x ticks={94}, try min ticks=9,
             xmin=\nmin, xmax=\nmax, ymin=\Cmin, ymax=\Cmax,
             xlabel style={at={(.5,-0.1)},anchor=east},
             ylabel style={at={(-0.1,0.5)},anchor=east},
             xlabel={\(n\)},
             ylabel={\(\log_2(\mathcal{C})\)},
             legend style={ at={(1.2,0.5)}, anchor=west },
             title={Complexity for $r=3, r=4, r=5$ in the overdetermined cases}
             ]
            \addplot[purple] table {\dataTheoTrois};
            \addlegendentry{Theoretical $r=3$}
            \addplot[violet] table {\dataTheoQuatre};
            \addlegendentry{Theoretical $r=4$}
            \addplot[magenta] table {\dataTheoCinq};
            \addlegendentry{Theoretical $r=5$}
            \addplot[teal,mark=+] table {\dataExpOptCinq};
            \addlegendentry{Experimental  $r=5$}
            \addplot[blue,mark=+] table {\dataExpOptQuatre};
            \addlegendentry{Experimental  $r=4$}
            \addplot[cyan,mark=+] table {\dataExpOptTrois};
            \addlegendentry{Experimental  $r=3$}

            \draw (axis cs:94,\Cmin) -- (axis cs:94,\Cmax)
            node[near start, left]{ROLLO-I};
            \draw (axis cs:36,\Cmin) -- (axis cs:36,\Cmax)
            node[very near end, right]{$k=\frac n 2$, $r=4$};
            \draw (axis cs:58,\Cmin) -- (axis cs:58,62)
            node[very near start, right]{$k=\frac n 2$, $r=5$};
          \end{axis}

          \end{tikzpicture}
          \caption{\label{fig:r=345exp}Theoretical vs Experimental value of the complexity of the
          computation.  The computations are done using \texttt{magma
            v2.22-2} on a machine with an
          Intel\textsuperscript{\textregistered}
          Xeon\textsuperscript{\textregistered} 2.00GHz processor (\emph{Any mention of commercial products is for information only and does not imply endorsement by NIST}). We
          measure the experimental complexity in terms of clock cycles of the CPU, given
          by the \texttt{magma} function \texttt{ClockCycles()}. The
          theoretical value is the binary logarithm of
          $m\binom{n-k-1}r\binom n {r}^{2.81-1}$. 
%JP The experimental values are the       binary logarithms of the aforementionned experimental complexity 
          $m$ is the largest prime
          less than $n/1.18$, $k$ is the minimum of $n/2$ (right part of
          the graph) and the largest value for which the system is
          overdetermined (left part).}
        \end{figure}

        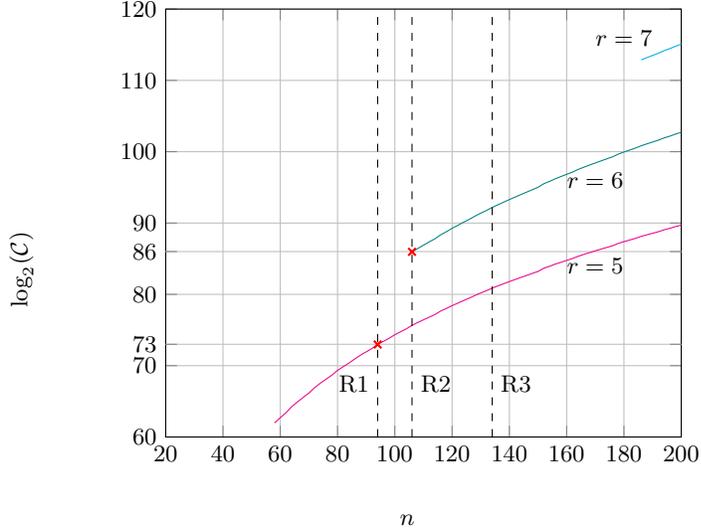
\begin{figure}[h]
          \centering
          \begin{tikzpicture}[scale=1]
           \pgfplotsset{every axis legend/.append style={at={(1,0.2)},anchor=east}}
           \def\nmin{20}\def\nmax{200}
           \def\Cmin{60}\def\Cmax{120}
           \begin{axis}[ grid=major, try min ticks=9, extra y ticks={73,86},
             xmin=\nmin, xmax=\nmax, ymin=\Cmin, ymax=\Cmax,
             xlabel style={at={(.5,-0.1)},anchor=east},
             ylabel style={at={(-0.1,0.5)},anchor=east},
             xlabel={\(n\)},
             ylabel={\(\log_2(\mathcal{C})\)},
             legend style={ at={(0.0,-0.4)}, anchor=west },
             title={Theoretical complexity for $r=5, 6, 7$ in the \emph{overdetermined} cases when $n=2k$.}
             ]
             \addplot[magenta] table {\dataTheoCinqKN};
             \addplot[teal] table {\dataTheoSixKN};
             \addplot[cyan] table {\dataTheoSeptKN};
             \draw (axis cs:170,84) node {$r=5$};
             \draw (axis cs:170,96) node {$r=6$};
             \draw (axis cs:180,116) node {$r=7$};

            \draw[dashed] (axis cs:94,\Cmin) -- (axis cs:94,\Cmax)
            node[very near start, left]{R1};
            \draw[dashed] (axis cs:106,\Cmin) -- (axis cs:106,\Cmax)
            node[very near start, right]{R2};
            \draw[dashed] (axis cs:134,\Cmin) -- (axis cs:134,\Cmax)
            node[very near start, right]{R3};
            \addplot[only marks,mark=x,red,thick] coordinates{ (94,72.97) (106,85.96)};
          \end{axis}
          \end{tikzpicture}
          \caption{\label{fig:theoretical-r=567}Theoretical value of the complexity of the computation in
          the overdetermined cases, which is the binary logarithm of
          $m\binom{n-k-1}r\binom n {r}^{2.81-1}$.  $m$ is the largest prime
          less than $n/1.18$, $n=2k$. The axis ``R1, R2, R3'' correspond to
          the values of $n$ for the cryptosystems ROLLO-I-128; ROLLO-I-192
          and ROLLO-I-256.}
        \end{figure}

        \begin{figure}[h]
          \centering
          \begin{tikzpicture}[scale=1]
           \pgfplotsset{every axis legend/.append style={at={(1,0.2)},anchor=east}}
           \def\nmin{20}\def\nmax{250}
           \def\Cmin{60}\def\Cmax{280}
           \begin{axis}[ grid=major, try min ticks=10, %extra y ticks={70,86,158},
             xmin=\nmin, xmax=\nmax, ymin=\Cmin, ymax=\Cmax,
             xlabel style={at={(.5,0)},anchor=east},
             ylabel style={at={(0,0.5)},anchor=east},
             xlabel={\(n\)},
             ylabel={\(\log_2(\mathcal{C})\)},
             title={Theoretical complexity for $r=5\dots 9$  when $n=2k$.}
             ]
             \addplot[magenta] table {\dataTheoOptCinq};
             \addplot[teal] table {\dataTheoOptSix};
             \addplot[cyan] table {\dataTheoOptSept};
             \addplot[blue] table {\dataTheoOptHuit};
             \addplot[red] table {\dataTheoOptNeuf};
             \draw (axis cs:170,84) node {$r=5$};
             \draw (axis cs:170,96) node {$r=6$};
             \draw (axis cs:80,150) node {$r=7$};
             \draw (axis cs:180,240) node {$r=8$};
             \draw (axis cs:45,240) node {$r=9$};

            \draw[dashed] (axis cs:94,\Cmin) -- (axis cs:94,\Cmax);
          %      node[very near start, left]{R1};
            \draw[dashed] (axis cs:106,\Cmin) -- (axis cs:106,\Cmax);
          %      node[very near start, right]{R2};
            \draw[dashed] (axis cs:134,\Cmin) -- (axis cs:134,\Cmax);
          %      node[very near start, right]{R3};
            \addplot[only marks,mark=x,red,thick] coordinates{ (94,70.15) (106,86.17) (134,157.9)};
          \end{axis}
          \end{tikzpicture}
          \caption{\label{fig:opt-r=5-9}Theoretical value of the complexity of
            RD in the overdetermined case (using punctured codes or specialization). $\mathcal{C}$  is 
            the smallest value between~\eqref{eq:complexitya}
            and~\eqref{eq:complexityp}.  $m$ is the largest prime less than
            $n/1.18$, $n=2k$. The dashed axes correspond to the values of $n$
            for the cryptosystems ROLLO-I-128; ROLLO-I-192 and ROLLO-I-256.}
        \end{figure}
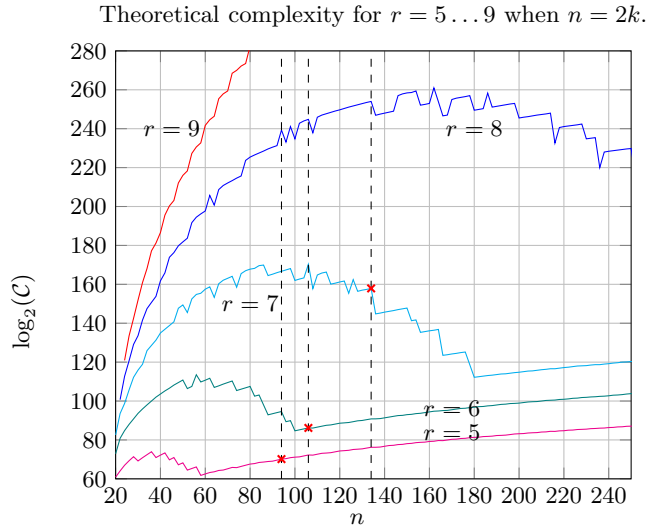

         Figure~\ref{fig:theoretical-r=567} shows the theoretical complexity for the same parameter regime
as Figure~\ref{fig:r=345exp} which fit the overdetermined case.
%         when $n=2k$ and $m$ is prime and close to $n/1.18$. We
%         take those parameters because they fit with the parameters in the
%         cryptosystem ROLLO-I. When the parameters $(m,n,k,r)$ do not satisfy
%         the condition $m\binom{n-k-1}r\ge \binom n r -1$, we
%         do not give the complexity. 
The graph starts from the first value of
         $n$ where $(n/1.18, n, 2k, r)$ is in the overdetermined case. We can
         see that theoretically, the cryptosystem ROLLO-I-128 with parameters
         $(79, 94, 47, 5)$ needs $2^{73}$ bit operations to decode an error,
         instead of the announced $2^{128}$ bits of security. In the same way,
         ROLLO-I-192 with parameters $(89, 106, 53, 6)$ would have 86 bits of
         security instead of 192. The parameters $(113, 134, 67, 7)$ for
         ROLLO-I-256 are not in the overdetermined case.

         There are two classical improvements that can be used to
         lower the complexity of solving an algebraic system. The
         first one consists in selecting a subset of all equations,
         when some of them are redundant, see~\cref{improve_over}. The
         second one is the hybrid attack that will be explained
         in~\cref{hybrid_over}.
        \subsection{Improvement in the ``super''-overdetermined case\label{improve_over} by puncturing}
        We consider the case when the system is
        ``super''-overdetermined, i.e. when the number of rows in $\MaxMin$ is
        really larger than the number of columns. In that case, it is not
        necessary to consider all equations, we just need the minimum number
        of them to be able to find the solution.

        To select the good equations (i.e. the ones that are likely to
        be linearly independent), we can take the system $\MaxMinors$
        obtained by considering code $\Cy$ punctured on the $p$ last
        coordinates, instead of the entire code. Puncturing code $\Cy$
        is equivalent to shortening the dual code, i.e. considering
        the system
        \begin{equation}
          \label{eq:punctMaxMinors}
          \MaxMinors\left(
            \Cm_{\any,\{1..n-p\}} (\trsp{\Hm})_{\{1..n-p\},\{1..n-k-1-p\}}
          \right).
        \end{equation}
        as we take $\Hm$ to be systematic on the last coordinates.  This
        system is formed by a sub-sequence of polynomials in $\MaxMinors$ that
        do not contain the variables $c_{i,j}$ with $n-p+1\le j \le n$. This
        system contains $m\binom{n-p-k-1}{r}$ equations in $\binom{n-p}{r}$
        variables $c_{T}$ with $ T \subset \{1..n-p-k-1\}$. If
        we take the maximal value of $p$ such that
        $m\binom{n-p-k-1}{r}\ge \binom{n-p}{r}-1$, we can still apply
        Algorithm~\ref{alg:overdetermined-case} but the complexity is reduced
        % for instance
        to
         \begin{equation}\label{eq:complexityp}
           \mathcal{O}\left(   m\binom{n-p-k-1}r {\binom {n-p} r}^{\omega-1}\right)
        \end{equation}
        operations in the field $\ff q$.

        \subsection{Reducing to the overdetermined case: hybrid attack\label{hybrid_over}}

        Another classical improvement consists in using an hybrid approach
        mixing exhaustive search and linear resolution, like
        in~\cite{BFP09}. This consists in specializing some variables of the
        system to reduce an underdetermined case to an overdetermined one.
For instance, if we specialize $a$ columns of the matrix $\Cm$, we are
        left with solving $q^{ar}$ linear systems $\MaxMin$ of size
        $m\binom{n-k-1}{r} \times \binom{n-a}r$, and the global cost is
         \begin{equation}\label{eq:complexitya}
           \mathcal{O}\left( q^{ar}  m\binom{n-k-1}r {\binom {n-a} r}^{\omega-1}\right)
        \end{equation}
        operations in the field $\ff q$. 
        In order to minimize the previous complexity (\ref{eq:complexitya}), one 
        chooses $a$ to be the smallest integer such that the condition 
        $m\binom{n-k-1}{r} \ge \binom {n-a}{r} -1$ is fulfilled. 
        Figure~\ref{fig:opt-r=5-9} page~\pageref{fig:opt-r=5-9} % and~\ref{fig:opt-r=5-9R3}
        gives the best theoretical complexities obtained for $r=5\dots 9$ with
        the best values of $a$ and $p$, for $n=2k$. Table~\ref{tab:num} page~\pageref{tab:num} gives
        the complexities of our attack (column ``\textbf{This paper}'') for all the parameters in the ROLLO and RQC
        submissions to the NIST competition; for the sake of clarity, we give the previous complexity 
        from \cite{BBBGNRT19}.

\section{Solving RD and MinRank: underdetermined case\label{attack_2}}\label{sec:underdetermined}

This section analyzes the support minors modeling approach (Modeling \ref{mod:support_minors_modeling}).

      \subsection{Solving \eqref{eq:minrank-homogeneous} by direct linearization}

      % Consider a generic MinRank problem involving $K$ matrices of dimension $m\times n$ with a target rank of $r$, where the Rank Decomposition Modeling equations are given by:

      % \[
      % \Sm \Cm =  \sum_{i=1}^{K}{x_iM_i}.
      % \]

      % Consider the $m$ matrices of dimension $(r+1) \times n$ given by  $\Cm$ stacked with a row, $\vec{r_j} = \vec{\pi_j} \sum_{i=1}^{K}{x_iM_i}$, of  $\sum_{i=1}^{K}{x_iM_i}$ , where $\vec{\pi_j}$ is the row vector with only one 1 on  the $j$st column :
      % \[
      % \Cm'_j=\begin{pmatrix}
      % \vec{r_j} \\
      % \Cm
      % \end{pmatrix}.
      % \]
      % For any $\Sm$, $\Cm$, $x_i$ solving the Rank Decomposition Modeling form of the MinRank problem, we have that $\vec{r_j}$ is in the span of the rows of $\Cm$ and therefore each matrix $\Cm'_j$ has rank at most $r$. This allows us to set up a new modeling for the MinRank problem by setting the $(r+1)\times(r+1)$ minors of the matrices $\Cm'_j$ equal to zero. The resulting equations, of which there are $m\binom{n}{r+1}$ can be expressed via Cofactor expansion with respect to their first row. In this way they can be seen to be expressible as bilinear forms in the variables $x_i$ and the $r\times r$ minors of $\Cm$, i.e. the variables $c_T$. As there are $K\binom{n}{r}$ 
The number of monomials that can appear in Modeling \ref{mod:support_minors_modeling} is $K \binom{n}{r}$ whereas the number of equations is 
$m \binom{n}{r+1}$. When the solution space of \eqref{eq:minrank-homogeneous} is of dimension 1, we expect to solve it by direct linearization whenever:
      \begin{equation}
      \label{b=1case}
      m \binom{n}{r+1} \geq K \binom{n}{r}-1.
      \end{equation}
      We did a lot of experiments as explained in \cref{experiences_minrank}, and they suggest that it is the case. 
      \begin{remark}
      Note that, in what follows, the \cref{b=1case} will sometimes be referred as the ``$b=1$ case''. 
      \end{remark}

      \subsection{Solving Support Minors Modeling at a higher degree, $q>b$
      \label{solving_support_minors_higher_degree}}
    In the case where \cref{b=1case} does not hold we may produce a generalized version of Support Minors Modeling, multiplying the Support Minors Modeling equations by homogeneous degree $b-1$ monomials in the linear variables, resulting in a system of equations that are homogeneous degree 1 in the  variables $c_{T}$ and homogeneous degree $b$ in the variables $x_i$. The strategy will again be to linearize over monomials.
Unlike in the simpler $b=1$ case, for $b \geq 2$ we cannot assume that all $m \binom{n}{r+1}\binom{K+b-2}{b-1}$ equations we produce in this way are linearly independent up to the point where we can solve the system by linearization. In fact, we can construct explicit linear relations between the equations starting at $b=2.$ 

In this section, we will focus on the simpler $q>b$ case. We will deal
with the common $q=2$ case in~\cref{seq:q=2}.  There is however an
unavoidable complication which occurs whenever we consider $b\ge q$, $q\ne 2$.

We can construct linear relations between the equations  from determinantal identities involving maximal minors of matrices whose first rows are some of the $\vec{r_j}$'s concatenated with $\Cm$. For instance we may write the trivial identity for any subset $J$ of columns of size $r+2$:
      $$
      \minor{\substack{
      \vec{r_j} \\
      \vec{r_k}\\
      \Cm}}{\any,J}  +  \minor{\substack{
      \vec{r_k} \\
      \vec{r_j}\\
      \Cm}}{\any,J}=0.
      $$ 
      Notice that this gives trivially a relation between certain equations corresponding to $b=2$ since
      a cofactor expansion along the first row of $   \minor{\substack{
      \vec{r_j} \\
      \vec{r_k}\\
      \Cm}}{\any,J}$ shows that this maximal minor is indeed a linear combination of 
      terms which is the multiplication of a linear variable $x_i$ with a maximal minor of the matrix $\begin{pmatrix}
      \vec{r_k}\\
      \Cm
      \end{pmatrix}$ (in other words an equation corresponding to $b=2$). A similar result holds for 
      $  \minor{\substack{
      \vec{r_k} \\
      \vec{r_j}\\
      \Cm}}{\any,J}$ where a cofactor expansion along the first row yields terms formed by a linear variable $x_i$ multiplied by a maximal minor of the matrix $\begin{pmatrix}
      \vec{r_j}\\
      \Cm
      \end{pmatrix}$. 
      This result can be generalized by considering symmetric tensors $(S_{j_1,\cdots, j_r})_{\substack{1 \leq j_1 \leq m\\
      \cdots \\ 1 \leq j_r \leq m
      }}$ of dimension $m$ of rank $b \geq 2$ over $\Fq$. Recall that these are tensors that satisfy
      $$
      S_{j_1,\cdots,j_b} = S_{j_{\sigma(1)},\cdots,j_{\sigma(b)}}$$
       for any permutation $\sigma$ acting on $\{1..b\}$.
      This is a vector space that is clearly isomorphic to the space of homogeneous polynomials of degree $b$ in $y_1,\cdots,y_m$ over $\Fq$. The dimension of this space is therefore $\binom{m+b-1}{b}$. We namely have
      \begin{proposition}\label{pr:symmetric}
      For any symmetric tensor $(S_{j_1,\cdots, j_b})_{\substack{1 \leq j_1 \leq m\\
      \cdots \\ 1 \leq j_b \leq m
      }}$ of dimension $m$ of rank $b\geq 2$ over $\Fq$ and any subset $J$ of $\{1..n\}$ of size $r+b$, we have:
      $$
      \sum_{j_1=1}^m\cdots \sum_{j_b=1}^m S_{j_1,\cdots, j_b}  \minor{\substack{
      \;\;\vec{r_{j_1}}\;\; \\
      \cdots \\
      \vec{r_{j_b}}\\
      \Cm}}{\any,J}=0.
      $$
      \end{proposition}
\begin{proof}
Notice first that the maximal minor $\minor{\substack{
      \;\;\vec{r_{j_1}} \;\; \\
      \cdots \\
      \vec{r_{j_b}}\\
      \Cm}}{\any,J}$ is equal to $0$ whenever at least two of the $j_i$'s are equal. The left-hand sum reduces therefore to a sum of terms of the form $\sum_{\sigma \in S_b} S_{\sigma(j_1),\cdots,\sigma(j_b)} \minor{\substack{
      \;\vec{r_{\sigma(j_1)}}\; \\
      \cdots \\
      \vec{r_{\sigma(j_b)}}\\
      \Cm}}{\any,J}$ where all the $j_i$'s are different. Notice now that from the fact that $S$ is a symmetric tensor we have
\begin{eqnarray*}
\sum_{\sigma \in S_b} S_{\sigma(j_1),\cdots,\sigma(j_b)} \minor{\substack{
     \; \vec{r_{\sigma(j_1)}} \; \\
      \cdots \\
      \vec{r_{\sigma(j_b)}}\\
      \Cm}}{\any,J}&= & S_{j_1,\cdots,j_b} \sum_{\sigma \in S_b}  \minor{\substack{
      \;\vec{r_{\sigma(j_1)}} \;\\
      \cdots \\
      \vec{r_{\sigma(j_b)}}\\
      \Cm}}{\any,J}\\
      &=&0
\end{eqnarray*}
because the determinant is an alternating form and there as many odd and even permutations in the symmetric group of order $b$ when $b \geq 2$. $\qed$
\end{proof}

This proposition can be used to understand the dimension $\mathrm{D}$ of the space of linear equations we obtain after linearizing the equations we obtain for a certain $b$. For instance for $b=2$ we obtain $m \binom{n}{r+1} K$ linear equations (they are obtained by linearizing the equations resulting from multiplying all the equations of the support minors modeling by one of the $K$ linear variables). However as shown by Proposition \ref{pr:symmetric} all of these equations are not independent and we have
$\binom{n}{r+2}\binom{m+1}{2}$ linear relations coming from all relations of the kind
\begin{equation}\label{eq:smm2}
\sum_{j=1}^m \sum_{k=1}^m S_{j,k}  \minor{\substack{
      \;\;\vec{r_{j}}\;\; \\
      \vec{r_{k}}\\
      \Cm}}{\any,J}=0.
\end{equation}
In our experiments, these relations turnt out to be independent yielding that the dimension $\mathrm{D}$ of this space should not be greater than $ m \binom{n}{r+1} K - \binom{n}{r+2}\binom{m+1}{2}$. Experimentally, we observed that we indeed had
$$
\Exp = m \binom{n}{r+1} K - \binom{n}{r+2}\binom{m+1}{2}.
$$
For larger values of $b$ things get more complicated but again Proposition \ref{pr:symmetric} plays a key role here.
Consider for example the case $b=3$. We have in this case $m \binom{n}{r+1} \binom{K+1}{2}$ equations obtained by multiplying all the equations of the support minors modeling by monomials of degree $2$ in the linear variables. Again these equations are not all independent, there are $\binom{m+1}{2}\binom{n}{r+2}K$ equations obtained by mutiplying all the linear relations between the $b=2$ equations derived from \eqref{eq:smm2} by a linear variable, they are of the form
\begin{equation}\label{eq:smm21}
x_i \sum_{j=1}^m \sum_{k=1}^m S_{j,k}  \minor{\substack{
      \;\;\vec{r_{j}}\;\; \\
      \vec{r_{k}}\\
      \Cm}}{\any,J}=0.
\end{equation}
But all these linear relations are themselves not independent as can be checked by using Proposition \ref{pr:symmetric} with $b=3$, we namely have for any symmetric tensor 
%%%%% ancienne formule
% $S_{jkl}$
$S_{i,j,k}$ of rank $3$:
\begin{equation}\label{eq:smm3}
%%%%% ancienne formule
% \sum_{j=1}^m \sum_{k=1}^m S_{i,j,k}  \minor{\substack{
%       \;\;\vec{r_{i}}\;\; \\
%        \vec{r_{j}}\\
%       \vec{r_{k}}\\
%       \Cm}}{\any,J}=0.
\sum_{i=1}^m\sum_{j=1}^m \sum_{k=1}^m S_{i,j,k}  \minor{\substack{
      \;\;\vec{r_{i}}\;\; \\
       \vec{r_{j}}\\
      \vec{r_{k}}\\
      \Cm}}{\any,J}=0.
\end{equation}
This induces linear relations among the equations \eqref{eq:smm21}, as  can be verified by a cofactor expansion along the first row of the left-hand term of \eqref{eq:smm3} which yields an equation of the form
$$
\sum_{i=1}^m x_i \sum_{j=1}^m \sum_{k=1}^m S^i_{j,k}  \minor{\substack{
      \;\;\vec{r_{j}}\;\; \\
      \vec{r_{k}}\\
      \Cm}}{\any,J}=0
$$
where the $\Sm^i=(S^i_{j,k})_{\substack{1 \leq j \leq m\\ 1 \leq k \leq m}}$ are symmetric tensors of order $2$. We would then expect that the dimension of the set of linear equations obtained from \eqref{eq:smm21} is only 
$\binom{m+1}{2}\binom{n}{r+2}K - \binom{n}{r+3}\binom{m+2}{3}$ yielding an overall dimension $\mathrm{D}$ 
%of the linearized system 
of order 
$$
{\mathrm D}= m \binom{n}{r+1} \binom{K+1}{2} - \binom{m+1}{2}\binom{n}{r+2}K + \binom{n}{r+3}\binom{m+2}{3},
$$
which is precisely what we observe experimentally. 
      This argument extends also to higher values of $b$, so that, if linear relations of the form considered above are the only relevant linear relations, then the number of linearly independent equations available for linearization at a given value of $b$ is:
\begin{heuristic}
      \begin{equation}
      \label{Expb<r+2}
      \Exp = \sum_{i=1}^{b}(-1)^{i+1}\binom{n}{r+i}\binom{m+i-1}{i}\binom{K+b-i-1}{b-i}.
      \end{equation}
\end{heuristic}
Experimentally, we found this to be the case with overwhelming
probability (see~\cref{experiences_minrank}) with the only general
exceptions being:

      \begin{enumerate}
      \item When $\Exp$ exceeds the number of monomials for a smaller value of $b$, typically 1, the number of equations is observed to be equal to the number of monomials for all higher values of $b$ as well, even if $\Exp$ does not exceed the total number of monomials at these higher values of $b$.
      \item When the MinRank Problem has a nontrivial solution and cannot be solved at $b=1$, we find the maximum number of linearly independent equations is not the total number of monomials but is less by 1. This is expected, since when the underlying MinRank problem has a nontrivial solution, then the Support Minors Modeling equations have a 1 dimensional solution space.
      \item When $b\ge r+2$, the equations are not any more linearly
      independent, and we give an explanation in~\cref{sec:b>r+2}.
      \end{enumerate}
      In summary, in the general case $q>b$, the number of monomials is $\binom{n}{r}\binom{K+b-1}{b}$ and  we expect to be able to linearize at degree $b$ whenever $b<r+2$ and
      \begin{equation}
        \label{condition_q_minrank}
        \binom{n}{r}\binom{K+b-1}{b} - 1 \leq \sum_{i=1}^{b}(-1)^{i+1}\binom{n}{r+i}\binom{m+i-1}{i}\binom{K+b-i-1}{b-i}
      \end{equation}
      Note that, for $b=1$, we recover the result~\eqref{b=1case}. As this system is very sparse, with $K(r+1)$ monomials 
      per equation, one can solve it using Wiedemann algorithm \cite{wiedemann1986solving}; thus the complexity to solve 
      MinRank problem when $b<q$, $b<r+2$ is 
      \begin{equation}
        \label{complexity_q_minrank}
        \mathcal{O}\left(
        K(r+1)
        \left(\binom{n}{r}\binom{K+b-1}{b}\right)^2
        \right)
      \end{equation}
      where $b$ is the smallest positive integer so that the condition (\ref{condition_q_minrank}) is fulfilled.
      \subsection{The $q=2$ case\label{seq:q=2}}
      The same considerations apply in the $q=2$ case, but due to the field equations, $x_i^2=x_i$, for systems with $b\ge2$, a number of monomials will collapse to a lower degree. This results in a system which is no longer homogeneous. Thus, in this case it is most profitable to combine the equations obtained at a given value of $b$ with those produced using all smaller values of $b$. Similar considerations to the general case imply that as long as $b<r+2$ we will have
\begin{equation}
      \Exp = \sum_{j=1}^b\sum_{i=1}^{j}(-1)^{i+1}\binom{n}{r+i}\binom{m+i-1}{i}\binom{K}{j-i}.
      \end{equation}
      equations with which to linearize the $\sum_{j=1}^b\binom{n}{r}\binom{K}{j}$
      monomials that occur at a given value of $b$. We therefore expect to be able to solve by linearization when $b < r+2$ and $b$ is large enough that

      \begin{equation}
      \label{eq:bformula:q2}
      \sum_{j=1}^b\binom{n}{r}\binom{K}{j} - 1 \leq  \sum_{j=1}^b\sum_{i=1}^{j}(-1)^{i+1}\binom{n}{r+i}\binom{m+i-1}{i}\binom{K}{j-i}.
      \end{equation}
      Similarly to the general case for any $q$ described in the previous section, the complexity to solve 
      MinRank problem when $q=2$ and $b<r+2$ is 
      \begin{equation}
        \label{complexity_q2_minrank}
        \mathcal{O}\left(
        K(r+1)
        \left(    \sum_{j=1}^b\binom{n}{r}\binom{K}{j}    \right)^2
        \right)
      \end{equation}
      where $b$ is the smallest positive integer so that the condition (\ref{eq:bformula:q2}) is fulfilled.
\subsection{Toward the $b\ge r+2$ case\label{sec:b>r+2}}
      We can also construct additional nontrivial linear relations starting at $b=r+2$. The simplest example of this sort of linear 
relations occurs when $m > r+1$. Note that each of the Support Minors modeling equations at $b=1$ is bilinear in the $x_i$ variables and a subset consisting of $r+1$ of the variables $c_T$. Note also, that there are a total of $m$ equations derived from the same subset (one for each row of $ \sum_{i=0}^{K}{x_iM_i}$ .) Therefore, if we consider the Jacobian of the $b=1$ equations with respect to the variables  $c_T$, the $m$ equations involving only $r+1$ of the variables $c_T$ will form a submatrix with $m$ rows and only $r+1$ nonzero columns. 
      Using a Cramer-like formula, we can therefore construct 
      left kernel vectors for these equations; its coefficients would 
      be degree $r+1$ polynomials in the $x_i$ variables. 
      Multiplying the equations by this kernel vector will produce zero, because the $b=1$ equations are homogeneous, and multiplying equations from a bilinear system by a kernel vector of the Jacobian of that system cancels all the highest degree terms. This suggests that \cref{Expb<r+2} needs to be modified when we consider values of $b$ that are $r+2$ or greater. These additional linear relations do not appear to be relevant in the most interesting range of $b$ for attacks on any of the cryptosystems considered, however.
      \subsection{Improvements for Generic Minrank\label{improvements_minrank}}
      The two classical improvements~\cref{improve_over} in the
      ``super''-overdetermined cases the hybrid attack
      and~\cref{hybrid_over} can also apply for Generic Minrank.
      
      We can consider applying the Support Minors Modeling techniques
      to submatrices $\sum_{i=1}^K \Mm'_ix_i$ of
      $\sum_{i=1}^K \Mm_ix_i$. Note that if $\sum_{i=1}^K \Mm_ix_i$
      has rank $ \leq r$, so does
      $\sum_{i=1}^K \Mm'_ix_i$ , so assuming we have a unique solution
      $x_i$ to both systems of equations, it will be the
      same. Generically, we will keep a unique solution in the smaller
      system as long as the decoding problem has a unique solution,
      i.e. as long as the Gilbert-Varshamov bound $K \le (m-r)(n-r)$ is
      satisfied.

      We generally find that the most beneficial settings use matrices
      with all $m$ rows, but only $n' \leq n$ of the columns. This
      corresponds to a puncturing of the corresponding matrix code over $\Fq$.
      It is always beneficial for the attacker to reduce $n'$ to the
      minimum value allowing linearization at a given degree $b$,
      however, it can sometimes lead to an even lower complexity to
      reduce $n'$ further and solve at a higher degree $b$. 

      On the other side, we can run exhaustive search on $a$ variables
      $x_i$ in $\Fq$ and solve $q^a$ systems with a smaller value of
      $b$, so that the resulting complexity is smaller than solving
      directly the system with a higher value of
      $b$. This optimization is considered in 
      the attack against ROLLO-I-256 (see \cref{tab:num}); 
      more details about this example are given in 
      \cref{attacks_against_rank_decoding}.

\subsection{Experimental results for Generic Minrank\label{experiences_minrank}}
We verified experimentally that the value of $\Exp$ correctly predicts
the number of linearly independent polynomials. We constructed random
systems (with and without a solution) for $q=2,13$, with $m=7,8$,
$r=2,3$, $n=r+3,r+4,r+5$, $K=3,\ldots,20$.  Most of the time, the
number of linearly independent polynomials was as expected. For
$q=13$, we had a few number of non-generic systems (usually less than
$1\%$ over 1000 random samples), and only in square cases where the
matrices have a predicted rank equal to the number of columns. For
$q=2$ we had a higher probability of linear dependencies, due to the
fact that over small fields, random matrices have a non-trivial
probability to be non invertible. Anyway, as soon as the field is big
enough or the number $\Exp$ is large compared to the number of
columns, all our experiments succeeded over 1000 samples.

%For $q=13$, there was a small number of exceptions

      \subsection{Using Support Minors Modeling together with MaxMin for RD\label{complexity_q2_under}}
      Recall that from MaxMin, we obtain $m \binom {n-k-1}{r}$ homogeneous linear equations in the $c_T$'s. These can be used to produce additional equations over the same monomials as used for Support Minors Modeling with $K=m(k+1)$.
However here, unlike in the overdetermined case, it is not interesting to specialize the matrix $\Cm$. Indeed, in that case it is sufficient to assume that the first component of $\ev$ is nonzero, then we can specialize to $\trsp{(\any,0,\dots,0)}$ the  first column of $\Sm\Cm$. Now,~\cref{eq:minrank-homogeneous} gives $m-1$ linear equations involving only the $x_i$'s, that allows us to eliminate $m-1$ variables $x_i$'s from the system and reduces the number of linear variables to $K=mk+1$. We still expect a space of dimension $1$ for the $x^\alpha c_T$'s, and this will be usefull for the last step of the attack described in~\cref{sec:laststep}.

      When $q>b$, we multiply the equations from MaxMin by 
%homogeneous 
degree $b$ monomials in the $x_i$'s. When $q=2$, this can be done by multiplying the MaxMin equations by monomials of degree $b$ or less. All these considerations lead to 
%With all the arguments mentioned above and the experiments mentioned in \cref{experiences_minrank}, we can make 
a similar heuristic as Heuristic~\ref{heuristic:over-determined}, i.e. linearization is possible for $q>b$, $0<b<r+2$ when:
      \begin{multline}
        \binom{n}{r}\binom{mk+b}{b} - 1 \leq \\
        m \binom {n-k-1}{r} \binom{mk+b}{b} + \sum_{i=1}^{b}(-1)^{i+1}\binom{n}{r+i}\binom{m+i-1}{i}\binom{mk+b-i}{b-i},\nonumber
      \end{multline}
      and for $q=2$, $0<b<r+2$ whenever:
      \begin{equation}
        \label{condition_q2_under}
          A_b-1 \leq B_c+C_b
      \end{equation}
      where
      \begin{align*}
        A_b & := \sum_{j=1}^b\binom{n}{r}\binom{mk+1}{j}    \\
        B_b & := \sum_{j=1}^b\left( m \binom{n-k-1}{r}\binom{mk+1}{j}\right) \\
        C_b & := \sum_{j=1}^b\sum_{i=1}^{j}\left((-1)^{i+1}\binom{n}{r+i}\binom{m+i-1}{i}\binom{mk+1}{j-i}\right).
      \end{align*}
      For the latter, it leads to a complexity of 
      \begin{equation}
        \label{complexity_q2_under_without_wiedemann}
        \mathcal{O}\left( (B_b+C_b)A_b^{\omega-1} \right)
      \end{equation}
      where $b$ is the smallest positive integer so that the condition (\ref{condition_q2_under}) is fulfilled.
      This complexity formula correspond to solving a linear system with $A_b$ unknowns and $B_b+C_b$ equations, recall that $\omega$ is the constant of linear algebra. 

      For a large range of parameters, this system is particularly sparse, so one could take advantage of that to use Wiedemann algorithm \cite{wiedemann1986solving}.
      More precisely, for values of $m$, $n$, $r$ and $k$ of 
      ROLLO or RQC parameters 
      (see \cref{param_ROLLO_I_basic_algo} and \cref{param_rqc}) 
      for which the condition (\ref{condition_q2_under}) is 
      fulfilled, we typically find that $b\approx r$.  

      In this case, $B_b$ equations consist of $\binom{k+r+1}{r}$ monomials, $C_b$ equations consist of $(mk+1)(r+1)$ monomials, 
      and the total space of monomials is of size $A_b$. 
      The Wiedemann's algorithm complexity can be written in term of the average number of monomials per equation, in our case it is 
      \[    D_b := \frac{B_b\binom{k+r+1}{r}+C_b(mk+1)(r+1)}{B_b+C_b}.   \]
      Thus the linearized system at degree $b$ is sufficiently sparse that Wiedemann outperforms Strassen for $b\geq 2$.  
      Therefore the complexity of support minors modeling bootstrapping MaxMin for RD is 
      \begin{equation}
        \label{complexity_q2_under_wiedemann}
        \mathcal{O}\left(
        D_b
        A_b^2
        \right)
      \end{equation}
      where $b$ is still the smallest positive integer so that the condition (\ref{condition_q2_under}) is fulfilled. A similar formula applies for the case $q>b$.
%\mb{replaced $        \left(\sum_{j=1}^b\binom{n}{r}\binom{mk+1}{j}\right)^2$ by $A_b^2$}
      \subsection{Last step of the attack\label{sec:laststep}}
      To end the attack on MinRank using Support Minors modeling or
      the attack on RD using MaxMinors modeling in conjunction with
      Support Minors modeling, one needs to find the value of each
      unknown.  When direct linearization at degree $b$ works, we get
      $\vv = (v_{\alpha,T}^*)_{\alpha,T}$ one nonzero vector
      containing one possible value for all $x^\alpha c_T$, where the
      $x^\alpha$'s are monomials of degree $b-1$ in the $x_i$'s, and all the
      other solutions are multiples of $\vv$ (as the solution space has
      dimension 1).

      In order to extract the values of all the $x_i$'s and thus finish the attack, one needs to find one $i_0$ and one $T_0$ such that $x_{i_0}\ne 0$ and $c_{T_0}\ne 0$. This is easily done by looking for a nonzero entry $v^*$ of $\vv$ corresponding to a monomial $x_{i_0}^{b-1}c_{T_0}$. At this point, we know that there is a solution of the system with $x_{i_0}=1$ and $c_{T_0}=v^*$.
      Then by computing the quotients of the entries in $\vv$ corresponding to the monomials $x_ix_{i_0}^{b-2}c_{T_0}$ and $x_{i_0}^{b-1}c_{T_0}$ we get the values of 
      \begin{equation}  
        \label{resolution_finale}
        x_i \ = \ \frac{x_i}{x_{i_0}}\ = \ \frac{x_ix_{i_0}^{b-2}c_{T_0}}{x_{i_0}^{b-1}c_{T_0}}, \quad 1\le i \le K.
      \end{equation}
      Doing so, one gets the values of all the $x_i$'s. This finishes the attack. This works without any assumption on MinRank, and with the assumption that the first coordinate of $\ev$ is nonzero for RD. If it is not the case, one uses another coordinate of $\ev$.
\section{Complexity of the attacks for different cryptosystems 
and comparison with generic Gr\"obner basis approaches
\label{4_experi_and_theory}}

        \subsection{Attacks against the Rank Decoding problem}
        \label{attacks_against_rank_decoding}
        \cref{tab:num} presents the best complexity of our attacks (see
        Sections \ref{attack_1} and \ref{attack_2}) against RD and
        gives the binary logarithm of the complexities (column \textbf{``This
        paper''}) for all the parameters in the ROLLO and RQC
        submissions to the NIST competition and Loidreau cryptosystem
        \cite{L17}; for the sake of clarity, we give the previous best
        known complexity from \cite{BBBGNRT19} (last column).  The
        third column gives the original rate for being
        overdeterminate.  The column `$a$' indicates the number of
        specialized columns in the hybrid approach
        (\cref{hybrid_over}), when the system is not
        overdetermined. Column `$p$' indicates the number of punctured
        columns in the ``super''-overdetermined cases
        (\cref{improve_over}). Column `$b$' indicates that we use
        Support Minors Modeling in conjunction with MaxMin
        (\cref{complexity_q2_under}).

        Let us give more details on how we compute the best complexity, 
        in \cref{tab:num}, for ROLLO-I-256 whose parameters are $(m,n,k,r)=(113,134,67,7)$. 
        The attack from~\cref{attack_1}
        only works with the hybrid approach, thus requiring $a=8$ and
        resulting in a complexity of $158$ bits
        (using~\eqref{eq:complexitya} and $\omega=2.81$). On the other
        hand, the attack from~\cref{complexity_q2_under} needs 
        $b=2$ which results in a complexity of $154$ (this time using 
        Wiedemann's algorithm). However, if we specialize $a=3$
        columns in $\Cm$, we get $b=1$ and the resulting complexity
        using Wiedemann's algorithm is $151$.
      % m:=113;n:=134;k:=67;r:=7; q:=2;
      % for a from 0 to 10 do a,m*binomial(n-k-1,r)-binomial(n-a,r)-1; od; 
      % a:=8; evalf(a*r + log(m*binomial(n-k-1,r))/log(2) + (2.81-1)*log(binomial(n-a,r))/log(2)); 
      % evalf(a*r + log(m*binomial(n-k-1,r))/log(2) + log(binomial(n-a,r))/log(2) + log(binomial(k+r+1,r))/log(2)); 
      %   for b from 1 to 2 do
      % Ab := binomial(n,r)*sum(binomial(m*k+1,jj), jj=1..b); Bb:=m*binomial(n-k-1,r)*sum( binomial(m*k+1,jj),jj=1..b);
      % Cb := 0:
      % for j from 1 to b do
      % Cb := Cb + sum((-1)^(ii+1)*binomial(n,r+ii)*binomial(m+ii-1,ii)*binomial(m*k+1,j-ii),ii=1..j);
      % od:
      % Cb;
      %   b, Ab-1 - (Bb+Cb);od;
        % b:=2; evalf(log(Bb+Cb)/log(2)+(2.81-1)*log(Ab)/log(2));
        %evalf(log(Bb*binomial(k+r+1,r)+Cb*(m*k+1)*(r+1))/log(2)-log(Bb+Cb)/log(2)+2*log(Ab)/log(2));
        \begin{table}[h]
          \caption{\label{tab:num}Complexity of the attack against
            RD for different cryptosystems.  
            A ``*'' in column ``\textbf{This paper}'' means that the best complexity uses 
Widemann's algorithm, otherwise Strassen's algorithm is used.
%The values in
%            the column \textbf{This paper} are the smallest ones
%            between Strassen's and Wiedemann's algorithm, the ``*''
%            indicates that it is Wiedemann.
}
          \centering
          \begin{tabular}{|c|c|c|c|c|c|c|c|}
            \hline & \textbf{$(m,n,k,r)$}
            &  $\frac{m\binom{n-k-1}r}{\binom{n}r-1}$  &$a$&$p$ & $b$ 
            & \textbf{This paper}  & \cite{BBBGNRT19}\\\hline
Loidreau (\cite{L17})  & $(128,120,80,4)$  & 1.28      & 0   & 43         & 0 &\textbf{65}    &98   \\\hline\hline
ROLLO-I-128  & $(79,94,47,5)$  & 1.97     &0&9                            & 0 &\textbf{71}    &117  \\\hline
ROLLO-I-192  & $(89,106,53,6)$  & 1.06    &0&0                            & 0 &\textbf{87}    &144  \\\hline
ROLLO-I-256  & $(113,134,67,7)$ &  0.67      &3& 0                        & 1 &\textbf{151*}  &197  \\\hline\hline
ROLLO-II-128 & $(83,298,149,5)$ & 2.42&0&40                               & 0 &\textbf{93}    &134  \\\hline
ROLLO-II-192  & $(107,302,151,6)$ & 1.53&0&18                             & 0 &\textbf{111}   &164  \\\hline
ROLLO-II-256  & $(127,314,157,7)$ & 0.89 &0&6                             & 1 &\textbf{159*}  &217  \\\hline\hline
ROLLO-III-128 & $(101,94,47,5)$  & 2.52&0&12                              & 0 &\textbf{70}    &119  \\\hline
ROLLO-III-192 & $(107,118,59,6)$ & 1.31&0&4                               & 0 &\textbf{88}    &148  \\\hline
ROLLO-III-256 & $(131,134,67,7)$ & 0.78&0&0                               & 1 &\textbf{131*}  &200  \\\hline\hline
RQC-I     & $(97,134,67,5)$  & 2.60&0&18                                  & 0 &\textbf{77}    &123  \\\hline
RQC-II    & $(107,202,101,6)$ &1.46&0&10                                  & 0 &\textbf{101}   &156  \\\hline
RQC-III   & $(137,262,131,7)$ & 0.93&3&0                                  & 0 &\textbf{144}   &214  \\\hline
         \end{tabular}
        \end{table}

        \subsection{Attacks against the MinRank problem}

        Table \ref{tab:GeMSS_Rainbow} shows the complexity of our attack against generic MinRank problems for GeMSS and Rainbow, two cryptosystems at the second round of the 
%JP aforementioned 
NIST competition. Our new attack is compared to the previous MinRank attacks, which use minors modeling in the case of GeMSS \cite{CFMPPR19}, and a linear algebra search \cite{DCPSY19} in the case of Rainbow. Concerning Rainbow, the acronyms  RBS and DA stand from Rainbow Band Separation and Direct Algebraic respectively; the column ``Best/Type'' shows the complexity of the previous best attack against Rainbow, which was not based on MinRank before our new attack (except for Ia). 
        All new complexities are computed by finding the number of columns $n'$ and the degree $b$ that minimize the complexity, as described in Section \ref{sec:underdetermined}. 

        \begin{table}
        \centering
        \caption{\label{tab:GeMSS_Rainbow}Complexity comparison between the new and the previous attacks against GeMSS and Rainbow parameters. 
        \cite{DCPSY19}.}
        \begin{tabular}{|c|c|c|c||c|c||c|c|c|}
        \hline
        &&&&&& \multicolumn{3}{c|}{\textbf{Complexity}}\\ \hline
        \textbf{GeMSS}$(D,n,\Delta,v)$ & $n/m$ & $K$  & $r$ & $n'$ & $b$ & New & Previous & Type \\ \hline
        GeMSS128$(513, 174, 12, 12) $     & 174 & 162 & 34  & 61    & 2 & \textbf{154}  & 522     & MinRank \\ \hline
        GeMSS192$(513, 256, 22, 20) $     & 265 & 243 & 52  & 94    & 2 & \textbf{223}  & 537     & MinRank\\ \hline
        GeMSS256$(513, 354, 30, 33) $     & 354 & 324 & 73  & 126   & 3 & \textbf{299}  & 1254    & MinRank\\ \hline\hline
        RedGeMSS128$(17, 177,15, 15) $    & 177 & 162 & 35  & 62    & 2 & \textbf{156}  & 538     & MinRank\\ \hline
        RedGeMSS192$(17, 266, 23, 25) $   & 266 & 243 & 53  & 95    & 2 & \textbf{224}  & 870     & MinRank\\ \hline
        RedGeMSS256$(17, 358, 34, 35) $   & 358 & 324 & 74  & 127   & 3 & \textbf{301}  & 1273    & MinRank\\ \hline\hline
        BlueGeMSS128$(129, 175, 13,14) $  & 175 & 162 & 35  & 63    & 2 & \textbf{158}  & 537     & MinRank\\ \hline
        BlueGeMSS192$(129,265, 22,23) $   & 265 & 243 & 53  & 95    & 2 & \textbf{224}  & 870     & MinRank\\ \hline
        BlueGeMSS256$(129, 358, 34, 32)$  & 358 & 324 & 74  & 127   & 3 & \textbf{301}  & 1273    & MinRank\\ \hline
          \multicolumn{9}{c}{}\\\hline
        \textbf{Rainbow}$(GF(q), v_1,o_1,o_2)$ & $n$ & $K$ & $r$ & $n'$ & $b$ & New & Previous & Best / Type \\ \hline 
        Ia$(GF(16),32,32,32)$ & 96 & 33 & 64            & 82  & 3 & 155 & 161 & 145 / RBS \\ \hline
        IIIc$(GF(256),68,36,36)$ & 140 & 37 & 104       & 125 & 5 & \textbf{208} & 585 & 215 / DA\\ \hline
        Vc$(GF(256),92,48,48)$ & 188 & 49 & 140         & 169 & 5 & \textbf{272} & 778 & 275 / DA \\ \hline
        \end{tabular}
        \end{table}
        %
        % \begin{table}
        % \centering
        % \caption{\label{tab:Rainbow}Comparison between the new MinRank attack, the previous best MinRank attack using linear algebra search, and the best known attack for Rainbow. Here the acronyms  RBS and DA stand from Rainbow Band Separation and Direct Algebraic, respectively \cite{DCPSY19}. The new complexity is computed by finding the number of columns $n'$ and the degree $b$ that minimizes the complexity, as described in Section \ref{sec:underdetermined}.} \label{tbl:cmplx:rainbow}
        % \begin{tabular}{|c|c|c|c||c|c||c|c|c|}
        % \hline
        % &&&&&& \multicolumn{3}{c|}{Complexity}\\ \hline
        % Rainbow$(GF(q), v_1,o_1,o_2)$ & $n$ & $K$ & $r$ & $n'$ & $b$ & New & Previous & Best / Type  \\ \hline
        % Ia$(GF(16),32,32,32)$ & 96 & 33 & 64            & 82  & 3 & 155 & 161 & 145/RBS \\ \hline
        % IIIc$(GF(256),68,36,36)$ & 140 & 37 & 104       & 125 & 5 & \textbf{208} & 585 & 215 / DA\\ \hline
        % Vc$(GF(256),92,48,48)$ & 188 & 49 & 140         & 169 & 5 & \textbf{272} & 778 & 275 / DA \\ \hline
        % \end{tabular}
        % \end{table}

        \subsection{Our approach vs. using 
%JP Comparison between our approach and the use of 
generic Gr\"obner basis algorithms}

Since our approach is an algebraic attack, it relies on solving a polynomial system, thus it 
%does 
looks like a Gr\"obner basis computation. 
In fact,
we do compute a Gr\"obner basis of the system, as we compute the unique
solution of the system, which represents its Gr\"obner basis. 

Nevertheless, our algorithm is not a generic Gr\"obner basis 
algorithm as it only works for the 
special type of system studied in this paper: the RD and MinRank
systems. As it is specifically designed for this purpose and for the 
reasons detailed below, it is more efficient than a generic algorithm. 

There are three main reasons why our approach is more efficient than
a generic Gr\"obner basis algorithm:
\begin{itemize}
\item[$\bullet$] We compute formally (that is to say at no extra cost except the 
size of the equations) new equations of degree r (the $\MaxMinors$ ones) 
that are already in the ideal, but not in the vector space
\[ 
  \mathcal{F}_r := 
  \langle uf : \text{$u$ monomial of degree $r-2$,  
  $f$ in the set of initial polynomials} \rangle.
\]
In fact, a careful analysis of a Gr\"obner basis computation with a
standard strategy shows that those equations are in 
$\mathcal{F}_{r+1}$, and that
the first degree fall for those systems is $r+1$. Here, 
we apply linear algebra directly on 
a small number of polynomials of degree
$r$ (see the next two items for more details), 
whereas a generic Gr\"obner basis algorithm would compute many polynomials
of degree $r+1$ and then reduce them in order to get those polynomials 
of degree $r$.
\item[$\bullet$] A classical Gr\"obner basis algorithm 
using linear algebra and a standard strategy typically constructs Macaulay matrices,
where the rows correspond to polynomials in the ideal and the columns
to monomials of a certain degree. Here, we introduce variables $c_T$
that represent maximal minors of $\mat{C}$, and thus represent not one
monomial of degree $r$, but $r!$ monomials of degree $r$. 
As we compute the Gr\"obner basis by using only polynomials that can be
expressed in terms of those variables (see the last item below), this 
reduces the number of columns of our matrices by a factor 
around $r!$ compared to generic Macaulay-like matrices.
\item[$\bullet$] The solution can be found by applying linear algebra 
only to some specific
equations, namely the MaxMinors ones in the overdetermined case, and
in the underdetermined case, equations that have degree 1 in the $c_T$
variables, and degree $b-1$ in the $x_i$ variables 
(see \cref{solving_support_minors_higher_degree}). 
This enables us to deal with polynomials involving only the $c_T$
variables and the $x_i$ variables, whereas a generic Gr\"obner basis algorithm would
consider all monomials up to degree $r+b$ in the $x_l$ and the $c_{i,j}$
variables. This drastically reduces the number of rows and columns in 
our matrices.
\end{itemize}

For all of those reasons, in the overdetermined case, only an elimination 
on our selected MaxMinors equations (with a ``compacted'' matrix with 
respect to the columns) 
is sufficient to get the solution; so we essentially avoid going 
up to the degree $r+1$ to produce those equations, 
we select a small number of rows, and gain a factor $r!$ on the number of columns. 

In the underdetermined case, we find linear equations by linearization
on some well-chosen subspaces of the vector space 
$\mathcal{F}_{r+b}$. We have
theoretical reasons to believe that our choice of subspace should lead 
to the computation of the solution (as usual, this is a ``genericity''
hypothesis), and it is confirmed by all our experiments. 

\section{Examples of new parameters for ROLLO-I and RQC}

        In light of the attacks presented in this article, it is possible to 
        give a few examples 
        of new parameters for the rank-based cryptosystems, submitted to the 
        NIST competition, ROLLO and RQC. With these new parameters, 
        ROLLO and RQC would be resistant to our attacks, while still 
        remaining attractive, for example with a loss of only 
        about 50 \% in terms of key size for ROLLO-I.

        For cryptographic purpose, parameters have to belong to an area which 
        does not correspond to the overdetermined case 
        and such that the hybrid approach would make the attack worse than in 
        the underdetermined case.

        Alongside the algebraic attacks in this paper, the best combinatorial attack against RD is 
        in \cite{AGHT18_sv}; its complexity for $(m,n,k,r)$ decoding is
%; as a reminder, for attacking a $[n,k]$ code over $\ff{q^m}$ with target rank r,
%        its complexity is 
        \[   \mathcal{O}\left((nm)^2 q^{r\Ceiling{\frac{m(k+1)}{n}}-m}\right).  \]
%JP ce n'est pas dans la section, c'est dans la table
%        \begin{remark}
%        In this section, the notation is the same as the one used in ROLLO and RQC 
%        submissions' specifications \acite{RQC2,ROLLO}. 
%        One should be careful that here, $n$ is the block-length and not the 
%        length of the code which can be either $2n$ or $3n$.
 %       \end{remark}
        In the following tables, we consider $\omega=2.81$. We also use the same notatin as in ROLLO and RQC 
       submissions' specifications \acite{RQC2,ROLLO}. In particular, $n$ is the block-length and not the 
        length of the code which can be either $2n$ or $3n$. Moreover, for ROLLO 
        (\cref{param_ROLLO_I_basic_algo}): 
        \begin{itemize}
          \item {\bf over/hybrid} is the cost of the hybrid attack; the value of $a$ is the smallest 
          to reach the overdetermined case, $a=0$ means that parameters are already in 
          the overdetermined case,
          \item {\bf under} is the case of underdetermined attack.
          \item {\bf comb} is the the cost of the best combinatorial attack mentioned above,
          \item {\bf DFR} is the binary logarithm of the 
          Decoding Failure Rate,
        \end{itemize} and for RQC (\cref{param_rqc}): 
        \begin{itemize}
          \item[$\bullet$] {\bf hyb2n(a)}: hybrid attack for length 
          $2n$, the value of $a$ is the smallest 
          to reach the overdetermined case, $a=0$ means that parameters are already in 
          the overdetermined case,
          \item[$\bullet$] {\bf hyb3n(a)}: non-homogeneous hybrid attack for length $3n$, $a$ is the same 
          as above. This attack corresponds to an adaptation of our 
          attack to a non-homogeneous error of the RQC scheme, 
          more details are given in \acite{RQC2},
          \item[$\bullet$] {\bf und2n}: underdetermined attack for length $2n$,
          \item[$\bullet$] {\bf comb3n}: combinatorial attack for length $3n$.
        \end{itemize}

 %       For more details about those parameters and the aforementioned 
 %       attacks, reader may refer to the submissions specifications of 
 %       ROLLO (see \acite{ROLLO}) 
 %       and RQC (see \acite{RQC2}).
        \begin{table}
          \centering
          \begin{tabular}{|l|c|c|c|c|c|c|c|c|c|c|c|c|c|}
          \hline
          Instance & $q$ & $n$ & $m$ & $r$ & $d$ & pk size (B) & DFR &  over/hybrid & $a$ & $p$ & under & $b$ &comb \\
          \hline
          {\tiny new2ROLLO-I-128} & 2 & 83 & 73 & 7 & 8 & 757 &-27 & 233 & 18& 0& 180 &3& 213 \\
          \hline
          {\tiny new2ROLLO-I-192} & 2 & 97 & 89 & 8 & 8 & 1057 & -33 & 258*& 17 & 0&197* &3& 283* \\
          \hline
          {\tiny new2ROLLO-I-256} & 2 & 113 & 103 & 9 & 9& 1454 & -33 & 408* & 30 &0&283* &6& 376* \\
          \hline
          \end{tabular}
          \caption{\label{param_ROLLO_I_basic_algo}New parameters and attacks complexities for ROLLO-I.}
        \end{table}
        \begin{table}
          \centering
          \begin{tabular}{|l|c|c|c|c|c|c|c|c|c|c|c|c|c|}
          \hline
          Instance & $q$ & $n$ & $m$ & $k$ & $w$ & $w_r$ & $\delta$ & pk (B) &  hyb2n(a) & hyb3n(a) & und2n& $b$ &comb3n \\
          \hline
          newRQC-I & 2 & 113 & 127 & 3 & 7 & 7 & 6 & 1793 & 160(6) & 211(0) & 158 & 1       & 205 \\
          \hline
          newRQC-II & 2 & 149 & 151 & 5 &  8 & 8 & 8 & 2812 & 331(24) & 262(0) & 224 &3     & 289 \\
          \hline
          newRQC-III & 2 & 179 & 181 & 3 & 9 & 9 & 7 & 4049 & 553(44) & 321(5) & 324 & 6    & 401 \\
          \hline

          \end{tabular}
          \caption{\label{param_rqc}New parameters 
          and attacks complexities for RQC.}
        \end{table}

\section{Conclusion}

      In this paper, we improve on the results by \cite{BBBGNRT19} on the Rank Decoding problem
      by providing a better
      analysis which permits to avoid the use of 
      generic Gr\"obner bases algorithms 
      and permits to 
      completely break rank-based cryptosystems parameters proposed to the 
      NIST Standardization Process,
      when the analysis in \cite{BBBGNRT19} only attacked them slightly.

      We generalize this approach to the case
      of the MinRank problem for which we obtain the best known complexity with algebraic attacks. 

      Overall, the results proposed in this paper give a new and deeper understanding 
      of the connections and the complexity of two problems of great interest in post-quantum cryptography: 
      the Rank Decoding and the MinRank problems. 

%%%%%%%% ancienne conclusion
      % In this paper we improve on the results by \cite{BBBGNRT19} on the Rank Decoding problem
      % by providing a better
      % analysis which permits to avoid the use of 
      % generic Gr\"obner basis algorithms 
      % and permits to 
      % completely break rank-based cryptosystems parameters proposed to the 
      % NIST Standardization Process,
      % when analysis in \cite{BBBGNRT19} only attacked slightly these parameters
      % (mostly corresponding to the overdeterminate case defined in \cite{BBBGNRT19}).

      % We generalize this approach to the case
      % of the MinRank problem for which we obtain the best known complexity with algebraic attacks. 
      % We also proposed a new approach for the underdeterminate case as described in 
      % \cite{BBBGNRT19}, for some parameters this attack supersedes the results of 
      % \cite{BBBGNRT19}, in particular for attacking ROLLO-I-256 parameters.

      % Overall the results proposed in this paper give a new and deeper understanding 
      % of the complexity of difficult problems based on the rank metric. These problems
      % have a strong interest since many systems still in the second round
      % of the NIST standardization process, like ROLLO, RQC, GeMSS or Rainbow can be attacked
      % through these problems.

\section*{Acknowledgements}

  We would like to warmly thank 
%address a special thank you to 
the reviewers, who did a wonderful 
  job by carefully reading our 
%long 
article and giving us useful feedback. 

  This work has been supported by the French ANR project CBCRYPT
  (ANR-17-CE39-0007) and the MOUSTIC project with the support from the
  European Regional Development Fund 
%(ERDF) 
and the Regional Council of
  Normandie.

  Javier Verbel was supported for this work by Colciencias scholarship 757 for PhD studies and the University of Louisville facilities. 
  
  We would like to thank John B Baena, and Karan Khathuria for useful discussions.
  We thank the Facultad de Ciencias of the Universidad Nacional de Colombia sede Medellín 
  for granting us access to the Enlace server, where we ran some of the experiments.

% BibTeX users should specify bibliography style 'splncs04'.
% References will then be sorted and formatted in the correct style.
\bibliographystyle{splncs04}
\bibliography{articlev2}

\begin{thebibliography}{10}
\providecommand{\url}[1]{\texttt{#1}}
\providecommand{\urlprefix}{URL }
\providecommand{\doi}[1]{https://doi.org/#1}

\bibitem{AABBBDGZCH19}
{Aguilar Melchor}, C., Aragon, N., Bettaieb, S., Bidoux, L., Blazy, O., Bros,
  M., Couvreur, A., Deneuville, J.C., Gaborit, P., Hauteville, A., Z{\'e}mor,
  G.: Rank quasi cyclic {(RQC)}. Second round submission to the NIST
  post-quantum cryptography call (Apr 2020)

\bibitem{AABBBDGHZ17}
Aguilar~Melchor, C., Aragon, N., Bettaieb, S., Bidoux, L., Blazy, O.,
  Deneuville, J.C., Gaborit, P., Hauteville, A., Z{\'e}mor, G.: {Ouroboros-R}.
  First round submission to the NIST post-quantum cryptography call (Nov 2017)

\bibitem{AABBBDGZ17}
{Aguilar Melchor}, C., Aragon, N., Bettaieb, S., Bidoux, L., Blazy, O.,
  Deneuville, J.C., Gaborit, P., Z{\'e}mor, G.: Rank quasi cyclic {(RQC)}.
  First round submission to the NIST post-quantum cryptography call (Nov 2017)

\bibitem{AGHT18_sv}
Aragon, N., Gaborit, P., Hauteville, A., Tillich, J.P.: A new algorithm for
  solving the rank syndrome decoding problem. In: Proc. IEEE ISIT (2018)

\bibitem{ABDGHRTZ17}
Aragon, N., Blazy, O., Deneuville, J.C., Gaborit, P., Hauteville, A., Ruatta,
  O., Tillich, J.P., Z{\'{e}}mor, G.: {LAKE} -- {L}ow r{A}nk parity check codes
  {K}ey {E}xchange. First round submission to the NIST post-quantum
  cryptography call (Nov 2017)

\bibitem{ABDGHRTZ17a}
Aragon, N., Blazy, O., Deneuville, J.C., Gaborit, P., Hauteville, A., Ruatta,
  O., Tillich, J.P., Z{\'{e}}mor, G.: {LOCKER} -- {LO}w rank parity {C}hec{K}
  codes {E}nc{R}yption. First round submission to the NIST post-quantum
  cryptography call (Nov 2017)

\bibitem{ABDGHRTZABBBO19}
Aragon, N., Blazy, O., Deneuville, J.C., Gaborit, P., Hauteville, A., Ruatta,
  O., Tillich, J.P., Z{\'{e}}mor, G., Aguilar~Melchor, C., Bettaieb, S.,
  Bidoux, L., Bardet, M., Otmani, A.: {ROLLO} (merger of {Rank-Ouroboros, LAKE
  and LOCKER}). Second round submission to the NIST post-quantum cryptography
  call (Apr 2020)

\bibitem{ABGHZ19}
Aragon, N., Blazy, O., Gaborit, P., Hauteville, A., Z{\'{e}}mor, G.: Durandal:
  A rank metric based signature scheme. In: Ishai, Y., Rijmen, V. (eds.)
  Advances in Cryptology -- EUROCRYPT 2019. pp. 728--758. Springer
  International Publishing, Cham (2019)

\bibitem{AGHRZ17}
Aragon, N., Gaborit, P., Hauteville, A., Ruatta, O., Z\'{e}mor, G.: Ranksign --
  a signature proposal for the {NIST's} call. First round submission to the
  NIST post-quantum cryptography call (Nov 2017)

\bibitem{AGHT18}
Aragon, N., Gaborit, P., Hauteville, A., Tillich, J.P.: A new algorithm for
  solving the rank syndrome decoding problem. In: 2018 {IEEE} International
  Symposium on Information Theory ({ISIT}). pp. 2421--2425. IEEE (2018)

\bibitem{BBBGNRT19}
{Bardet}, M., {Briaud}, P., {Bros}, M., {Gaborit}, P., {Neiger}, V., {Ruatta},
  O., {Tillich}, J.P.: {An Algebraic Attack on Rank Metric Code-Based
  Cryptosystems}. Advances in Cryptology - EUROCRYPT~2020  (May 2020)

\bibitem{BFP09}
Bettale, L., Faugere, J.C., Perret, L.: Hybrid approach for solving
  multivariate systems over finite fields. Journal of Mathematical Cryptology
  \textbf{3}(3),  177--197 (2009)

\bibitem{BFS99}
Buss, J.F., Frandsen, G.S., Shallit, J.O.: The computational complexity of some
  problems of linear algebra. J. Comput. System Sci.  \textbf{58}(3),  572--596
  (Jun 1999)

\bibitem{CFMPPR19}
Casanova, A., Faug{\`{e}}re, J., Macario-Rat, G., Patarin, J., Perret, L.,
  Ryckeghem, J.: {{GeMSS}: A Great Multivariate Short Signature}. Second round
  submission to the NIST post-quantum cryptography call (Apr 2019)

\bibitem{C01}
Courtois, N.: Efficient zero-knowledge authentication based on a linear algebra
  problem {MinRank}. In: Advances in Cryptology - ASIACRYPT~2001. LNCS,
  vol.~2248, pp. 402--421. Springer, Gold Coast, Australia (2001)

\bibitem{DT18b}
{Debris-Alazard}, T., Tillich, J.P.: Two attacks on rank metric code-based
  schemes: Ranksign and an identity-based-encryption scheme. In: Advances in
  Cryptology - ASIACRYPT~2018. pp. 62--92. LNCS, Springer, Brisbane, Australia
  (Dec 2018)

\bibitem{DCPSY19a}
Ding, J., Chen, M.S., Petzoldt, A., Schmidt, D., Yang, B.Y.: {Gui}. First round
  submission to the NIST post-quantum cryptography call (Nov 2017)

\bibitem{DCPSY19}
Ding, J., Chen, M.S., Petzoldt, A., Schmidt, D., Yang, B.Y.: Rainbow. Second
  round submission to the NIST post-quantum cryptography call (Apr 2019)

\bibitem{FLP08}
Faug{\`e}re, J.C., {Levy-dit-Vehel}, F., Perret, L.: Cryptanalysis of
  {Minrank}. In: Wagner, D. (ed.) Advances in Cryptology - CRYPTO~2008. LNCS,
  vol.~5157, pp. 280--296 (2008)

\bibitem{FSS10}
Faug{\`{e}}re, J., Safey El~Din, M., Spaenlehauer, P.: Computing loci of rank
  defects of linear matrices using {G}r{\"{o}}bner bases and applications to
  cryptology. In: International Symposium on Symbolic and Algebraic
  Computation, {ISSAC} 2010, Munich, Germany, July 25-28, 2010. pp. 257--264
  (2010)

\bibitem{G85}
Gabidulin, E.M.: Theory of codes with maximum rank distance. Problemy Peredachi
  Informatsii  \textbf{21}(1),  3--16 (1985)

\bibitem{GPT91}
Gabidulin, E.M., Paramonov, A.V., Tretjakov, O.V.: Ideals over a
  non-commutative ring and their applications to cryptography. In: Advances in
  Cryptology - EUROCRYPT'91. pp. 482--489. No.~547 in LNCS, Brighton (Apr 1991)

\bibitem{GMRZ13}
Gaborit, P., Murat, G., Ruatta, O., Z{\'e}mor, G.: Low rank parity check codes
  and their application to cryptography. In: Proceedings of the Workshop on
  Coding and Cryptography WCC'2013. Bergen, Norway (2013)

\bibitem{GRS16}
Gaborit, P., Ruatta, O., Schrek, J.: On the complexity of the rank syndrome
  decoding problem. IEEE Trans. Inform. Theory  \textbf{62}(2),  1006--1019
  (2016)

\bibitem{GRSZ14}
Gaborit, P., Ruatta, O., Schrek, J., Z{\'{e}}mor, G.: New results for
  rank-based cryptography. In: Progress in Cryptology - AFRICACRYPT~2014. LNCS,
  vol.~8469, pp. 1--12 (2014)

\bibitem{GRSZ14a}
Gaborit, P., Ruatta, O., Schrek, J., Z{\'{e}}mor, G.: Ranksign: An efficient
  signature algorithm based on the rank metric (extended version on arxiv). In:
  Post-Quantum Cryptography~2014. LNCS, vol.~8772, pp. 88--107. Springer (2014)

\bibitem{GZ14}
Gaborit, P., Z{\'{e}}mor, G.: On the hardness of the decoding and the minimum
  distance problems for rank codes. IEEE Trans. Inform. Theory
  \textbf{62(12)},  7245--7252 (2016)

\bibitem{HPS98}
Hoffstein, J., Pipher, J., Silverman, J.H.: {NTRU}: A ring-based public key
  cryptosystem. In: Buhler, J. (ed.) Algorithmic Number Theory, Third
  International Symposium, ANTS-III, Portland, Oregon, USA, June 21-25, 1998,
  Proceedings. LNCS, vol.~1423, pp. 267--288. Springer (1998)

\bibitem{KS99}
Kipnis, A., Shamir, A.: Cryptanalysis of the {HFE} public key cryptosystem by
  relinearization. In: Advances in Cryptology - CRYPTO'99. LNCS, vol.~1666, pp.
  19--30. Springer, Santa Barbara, California, USA (Aug 1999)

\bibitem{L17}
Loidreau, P.: A new rank metric codes based encryption scheme. In: Post-Quantum
  Cryptography~2017. LNCS, vol. 10346, pp. 3--17. Springer (2017)

\bibitem{MTSB12}
Misoczki, R., Tillich, J.P., Sendrier, N., Barreto, P.S.L.M.:
  {MDPC}-{McEliece}: New {McEliece} variants from moderate density parity-check
  codes (2012)

\bibitem{OTN18}
Otmani, A., Tal{\'{e}}-Kalachi, H., Ndjeya, S.: Improved cryptanalysis of rank
  metric schemes based on {G}abidulin codes. Des. Codes Cryptogr.
  \textbf{86}(9),  1983--1996 (2018)

\bibitem{OJ02}
Ourivski, A.V., Johansson, T.: New technique for decoding codes in the rank
  metric and its cryptography applications. Problems of Information
  Transmission  \textbf{38}(3),  237--246 (2002)

\bibitem{O05}
Overbeck, R.: A new structural attack for {GPT} and variants. In: Mycrypt.
  LNCS, vol.~3715, pp. 50--63 (2005)

\bibitem{P96}
Patarin, J.: Hidden fields equations {(HFE)} and isomorphisms of polynomials
  {(IP):} two new families of asymmetric algorithms. In: Maurer, U. (ed.)
  Advances in Cryptology --- EUROCRYPT '96. pp. 33--48. Springer Berlin
  Heidelberg, Berlin, Heidelberg (1996)

\bibitem{PCYTD15}
Petzoldt, A., Chen, M., Yang, B., Tao, C., Ding, J.: Design principles for
  {HFEv}- based multivariate signature schemes. In: Iwata, T., Cheon, J.H.
  (eds.) Advances in Cryptology -- ASIACRYPT 2015. pp. 311--334. Springer
  Berlin Heidelberg, Berlin, Heidelberg (2015)

\bibitem{PBD14}
Porras, J., Baena, J., Ding, J.: {ZHFE}, a new multivariate public key
  encryption scheme. In: Mosca, M. (ed.) Post-Quantum Cryptography - 6th
  International Workshop, PQCrypto 2014, Waterloo, ON, Canada, October 1-3,
  2014. Proceedings. LNCS, vol.~8772, pp. 229--245. Springer (2014)

\bibitem{VBCPS19}
Verbel, J., Baena, J., Cabarcas, D., Perlner, R., Smith{-}Tone, D.: On the
  complexity of ``superdetermined'' {Minrank} instances. In: Post-Quantum
  Cryptography~2019. LNCS, vol. 11505, pp. 167--186. Springer, Chongqing, China
  (May 2019)

\bibitem{wiedemann1986solving}
Wiedemann, D.: Solving sparse linear equations over finite fields. IEEE Trans.
  Inform. Theory  \textbf{32}(1),  54--62 (1986)

\end{thebibliography}

\end{document}